\documentclass[journal,onecolumn,12pt]{IEEEtran}
\usepackage[doublespacing,nodisplayskipstretch]{setspace}
\usepackage{amsmath,graphicx,mathtools}
\usepackage{amssymb,amsthm,verbatim,mathtools}
\usepackage[english]{babel}
\usepackage{csquotes}
\usepackage{bbm}
\usepackage{xcolor,tikz}
\usepackage{url,cite}
\usepackage[colorlinks,pdfstartview=FitH,citecolor=blue,linkcolor=blue,urlcolor=blue]{hyperref}
\usepackage{transparent}
\usepackage{mathrsfs}
\usepackage{enumitem}
\usepackage{subfig}
\usepackage[most]{tcolorbox} 
\usetikzlibrary{calc, bending, backgrounds}
\definecolor{lightgray}{RGB}{224,224,224}

\newtheorem{theorem}{Theorem}
\newtheorem{example}{Example}
\newtheorem{definition}{Definition}
\newtheorem{proposition}{Proposition}
\newtheorem{lemma}{Lemma}
\newtheorem{remark}{Remark}

\newcommand{\srr}{\mathcal{S}(\boldsymbol{G})}

\newcounter{myboxctr}
\newtcolorbox{myrefbox}[2][]{%
  colback= white,
  colframe=gray,
  fonttitle=\bfseries,
  before upper={},
  title={\refstepcounter{myboxctr}\label{#1}Box~\themyboxctr:~#2}
}

\author{Hoang~Ly 
        and Emina~Soljanin
\thanks{H.~Ly and E.~Soljanin are with the Department of Electrical and Computer Engineering, Rutgers University, New Jersey, USA. E-mail: \{\texttt{mh.ly;emina.soljanin}\}@rutgers.edu.}}

\begin{document}
\tikzset{every picture/.style={line width=1.15pt}}
\title{Service Rate Regions of MDS Codes and\\[-1ex] Fractional Matchings in Quasi-uniform Hypergraphs}

\maketitle

\author{%
  \IEEEauthorblockN{Hoang Ly and Emina Soljanin}
  \IEEEauthorblockA{
  Rutgers University\\
                    E-mail: \{\texttt{mh.ly,emina.soljanin}\}@rutgers.edu}
}

\maketitle
\begin{abstract}
The Service Rate Region (SRR) is a fundamental performance metric for distributed storage systems with redundancy, quantifying the range of data request patterns that the system can serve concurrently without exceeding the servers' capacities. Geometrically, the SRR is a convex polytope in $\mathbb{R}^k$, where each axis corresponds to the request rate for one of $k$ data objects. This paper focuses on the SRR of systems encoded using a large class of Maximum Distance Separable (MDS) codes. We characterize the axial intercepts of the SRR and derive the smallest standard simplex containing it, demonstrating that the SRR expands strictly as the number of systematic columns in the generator matrix of the code increases. To facilitate this analysis, we develop a graph-theoretic framework that models the SRR as the image of a fractional matching polytope of a quasi-uniform hypergraph. We introduce a novel allocation strategy, \emph{Greedy Matching}, and use it to characterize the SRRs without requiring an exhaustive evaluation of fractional matchings. These findings yield exact SRRs for a wide range of MDS codes and unify previous results regarding systematic and non-systematic constructions. Our work offers both theoretical and practical guidance for designing storage systems with optimized data service. 
\end{abstract}






%
\newpage
\section{Introduction}
\label{sec:Introduction}
Erasure-coded distributed storage systems provide reliability and availability by storing encoded data objects across multiple servers. Redundancy not only protects against node failures but also mitigates performance degradation in download times caused by stragglers (slow or unresponsive servers). Previous research has primarily focused on optimizing storage efficiency by minimizing redundancy overhead and improving recovery efficiency of lost data through reduced repair bandwidth or repair degree~\cite{DSS:journals/ftcit/RamkumarBSVKK22}. Another line of work focuses on the download-time performance evaluation of distributed storage systems \cite{Download:KadheSS15, Download:AktasKSS21}. Evaluating download time is challenging because coded systems introduce complex data-access queues \cite{Download:AktasNS17}. Even in simple cases, one has to resort to heuristics \cite{Download:AktasS18}. 
The stability regions of these queues, known as the \emph{service rate region} (SRR), have emerged as a fundamental metric for evaluating the performance of erasure-coded distributed storage systems~\cite{Service:journals/tit/AktasJKKS21}. 

For a system that stores \(k\) data objects redundantly on \(n\) servers, each capable of handling requests at a rate \(\mu\), we define the SRR as the set of request rate vectors \((\lambda_1, \dots, \lambda_k)\) that the system can serve concurrently without exceeding the capacity of any server. Geometrically, the SRR forms a convex polytope in \(\mathbb{R}^k\), with each axis corresponding to the request rate for a distinct data object~\cite{Service:conf/allerton/AktasAJJKMMS17}. Recent work has focused on characterizing SRRs under various coding schemes. The exact SRRs are known for Simplex codes~\cite{Service:conf/allerton/AktasAJJKMMS17}, Hamming codes~\cite{Service_Hamming,Service_Design:preprint/arxiv/LyS25}, certain structured MDS codes~\cite{Service:journals/tit/AktasJKKS21, Service:conf/isit/KazemiKSS20}, and for Reed-Muller codes of first-order and general (higher orders)~\cite{Service:conf/isit/KazemiKS20, SRR_RM:conf/isit, Service:preprint/arxiv/LySL25}. The general properties of SRR polytopes have been explored in~\cite{Service:journals/siaga/AlfaranoKRS24}, and connections to the majority-logic decoding and combinatorial design theory have also been established~\cite{Service_Design:preprint/arxiv/LyS25}. Parallel research has examined the code design problem: constructing codes whose SRRs meet specific performance or structural constraints~\cite{Parameters:conf/isit/KilicRS24}.

This paper characterizes the SRRs of a broad class of MDS codes, which remain central to both practical deployments and theoretical analyses of erasure-coded systems. Building on a graph-theoretic foundation, we extend prior frameworks to methodically analyze MDS codes with varying numbers of systematic (uncoded) and coded server nodes. Our approach leverages fractional matchings in hypergraphs to expose new structural insights into the geometry of SRRs.

Distributed storage allocation problems are often reformulated as hypergraph matching problems, where recovery sets correspond to hyperedges~\cite{allocations:journals/tcom/PengNS21, LargeMatching:journals/jct/AlonFHRRS12}. These formulations lend themselves to efficient relaxation via fractional matchings, enabling linear programming techniques for performance analysis and optimization. However, in practice, these problems often become computationally intractable due to the high degree of overlap among hyperedges or the inherent complexity of identifying all such recovery sets~\cite{Service:preprint/arxiv/LySL25}. Such structural challenges motivate the need for more refined combinatorial approaches to understanding and characterizing the underlying SRR.

We here leverage the concept of \emph{recovery hypergraphs}, first defined in~\cite{Service:conf/isit/KazemiKSS20}, to translate the problem of analyzing SRRs into an equivalent graph-theoretic problem. We derive explicit geometric inner and outer bounds for these regions, providing clear insights into how different storage strategies directly influence the SRR.
In particular, by analyzing these recovery hypergraphs, we prove that increasing the number of systematic servers enlarges the SRR, highlighting a fundamental trade-off between simplicity and flexibility in data service. Moreover, we introduce a novel data request allocation strategy called \emph{Greedy Matching} and establish its optimality, i.e., that it suffices to focus on these matchings to characterize the SRR rather than all possible fractional matchings. Utilizing this strategy, we explicitly characterize SRRs across various system configurations, generalizing and refining existing results. Our findings offer a rigorous foundation for understanding and optimizing data service capabilities in distributed storage systems, paving the way for superior system design and more efficient data access strategies.
\textbf{}

The remainder of the paper is organized as follows. Section~\ref{Sec:Problem_Formulate} first introduces the redundant storage model and formally defines service rate region (SRR) of distributed storage systems. It then reformulates the problem in graph-theoretic and linear programming terms and introduces relevant combinatorial concepts and tools. Section~\ref{Sec:Bounding_simplices} derives two geometric bounds for SRRs. Section~\ref{Sec:Inclusion} uses the derived bounds to establish an inclusion theorem for SRRs, proving that SRRs strictly expand as the number of systematic nodes increases. In Section~\ref{Sec:Greedy_Matching}, we propose a novel \emph{Greedy Matching} allocation strategy and prove its optimality. Section~\ref{Sec:SRR_polytopes} explicitly characterizes the SRRs in various configurations of the system. Finally, Section~\ref{Sec:Conclusion} concludes the paper and outlines potential directions for future research.

\section{Problem Formulation}\label{Sec:Problem_Formulate}





 We first describe a redundant storage model consisting of $n$ identical nodes (servers) that collectively store multiple data objects with redundancy. We then construct a family of underlying MDS \emph{generator matrices} that govern how data objects are encoded and distributed across these servers. Finally, we introduce certain concepts necessary for our analysis.
 
 When mentioning well-known concepts in the literature, we use \emph{italic}. We formally introduce the new and less standard concepts through a \emph{Definition}. The matrices and standard basis vectors are denoted \textbf{in bold}. The finite field over a prime power \(q\) is denoted by \(\mathbb{F}_q\). A linear code \(\mathcal{C}\) on \(\mathbb{F}_q\) with parameters \([n, k, d]_q\), is a $k$-dimensional subspace of the $n$-dimensional vector space \(\mathbb{F}_q^n\) with a minimum Hamming distance $d$. The symbols \(\boldsymbol{0}_k\) and \(\boldsymbol{1}_k\) denote the all-zero and all-one column vectors of length \(k\), respectively. The standard basis (column) vector with a one at position \(i\) and zero elsewhere is represented by \(\boldsymbol{e}_i\), and its transpose by $\boldsymbol{e}_i^{\top}$. The set of positive integers not exceeding \(i\), which is $\{1, 2, \dots, i\}$, is denoted \([i]\). Finally, $(x)^+ = \max\{x, 0\}$ for any real number $x$.

\subsection{Redundant Storage using MDS Codes}
Consider a storage system that stores \( k \ge 2 \) data objects on \( n \geq k \) servers (or, \emph{nodes}), labeled \( 1, \dots, n \). We assume that all data objects have the same size and that each server has a storage capacity of one object, that is, \emph{unit storage capacity}. The assumption allows us to mathematically represent objects as elements of some finite field \( \mathbb{F}_q \). 
Each server stores a linear combination of the data objects in $\mathbb{F}_q$. The system can therefore be specified by a matrix \( \boldsymbol{G} \, \in \, \mathbb{F}_q^{k \times n} \), called the {\it generator matrix}. If $\mathrm{o} = (o_1,\dots, o_k)$ is a row vector in $\mathbb{F}_q^n$ of data objects, then the $j$-th coordinate of the resultant vector $\mathrm{o}\cdot \boldsymbol{G}$, for $j \, \in \, \{1, 2, \dots, n\}$, is the coded object stored on the $j$-th server. Each column \( \boldsymbol{c}_i \), \( 1 \le i \le n \), of the generator matrix \( \boldsymbol{G} \) corresponds to server \( i \). If \( \boldsymbol{c}_i \) is equal to a standard basis vector, the server is called \emph{systematic} (or \emph{uncoded}); otherwise, it is referred to as \emph{coded}.

We construct a family of generator matrices for MDS codes over $\mathbb{F}_q$ with varying numbers of systematic columns. We aim to understand and quantify how these variations impact the SRR. The construction is as follows: start with a $k \times (k+n)$ MDS matrix $\boldsymbol{M}$ over a finite field $\mathbb{F}_q$ where $q$ is a prime or prime power such that $q \ge k+n+1$. Its structure is given by
\[
\boldsymbol{M} =
\Bigl[\boldsymbol{e}_1 \mid \boldsymbol{e}_2 \mid \dots \mid \boldsymbol{e}_k \mid p_1 \mid \dots \mid p_n\Bigr],
\]
where the first $k$ columns, $\boldsymbol{e}_j$, are standard basis vectors in $\mathbb{F}_q^k$. The remaining $n$ columns, $p_1, \dots, p_n$, are parity-check columns. Because its leftmost $k$ columns form the identity matrix $\boldsymbol{I}_k$, $\boldsymbol{M}$ is called systematic. The construction and existence of such systematic MDS matrices over $\mathbb{F}_q$ have been shown, for example, in~\cite{DSS:journals/ftcit/RamkumarBSVKK22,MDS:journals/ComLet/LacanF04}.

From $\boldsymbol{M}$, we construct a family of $k \times n$ matrices $\boldsymbol{G}_i(n, k)$, indexed by $i$ where $0 \le i \le k$. Each $\boldsymbol{G}_i(n, k)$ is obtained by selecting the $i$ leftmost systematic columns and the $n-i$ rightmost parity columns of $\boldsymbol{M}$:
\[
\boldsymbol{G}_i(n, k) \;=\;
\Bigl[\boldsymbol{e}_1 \mid \boldsymbol{e}_2 \mid \dots \mid \boldsymbol{e}_i \mid p_{i+1} \mid \dots \mid p_n\Bigr].
\] 
We will often write $\boldsymbol{G}_i$ instead of $\boldsymbol{G}_i(n,k)$, and denote by $G_i^l$ (without boldface) the $l$-th column of $\boldsymbol{G}_i$. Since $\boldsymbol{M}$ is an MDS matrix, any $k$ columns of $\boldsymbol{M}$ are linearly independent, and the same holds for any $k$ among the $n \ge k$ columns of $\boldsymbol{G}_i$. Hence each $\boldsymbol{G}_i$ also generates an MDS code. Moreover, $\boldsymbol{G}_i$ and $\boldsymbol{G}_{i+1}$ differ in exactly one column, namely, the $(i+1)$-th column. These observations will play a key role in proving the inclusion relations for service regions (see Section~\ref{Sec:Inclusion}).

The matrices \( \boldsymbol{G}_0 \) and \( \boldsymbol{G}_k \) correspond to two extremal cases: systems with no systematic nodes (i.e., all nodes are coded) and systems having all \( k \) systematic nodes, respectively. These boundary cases have been the primary focus of prior SRR studies for MDS codes~\cite{Service:journals/tit/AktasJKKS21}. In this work, we consider a more general setting in which the generator matrices have $i$ systematic nodes, for any \( 0 \le i \le k \), thereby interpolating between the fully coded and fully systematic regimes. Practically, in a storage system generated by \( \boldsymbol{G}_i \), each of the first \( i \) servers stores the uncoded copy of a single data object (raw data). The remaining \( n - i \) servers store linearly coded copies.

\subsection{Recovery Sets and Service Rate Region} 
A recovery set for object $o_i$ is a set of stored symbols that can be used to recover $o_i$. Let $\boldsymbol{c}_i$ denote the $i$-th column of the generator matrix $\boldsymbol{G}$. A set $R \subseteq [n]$ of column indices in $\boldsymbol{G}$ is a \emph{recovery set} for data object $o_j$ (aka recovery set for vector $\boldsymbol{e}_j$), if
\[
  \boldsymbol{e}_j \, \in \, \textsf{span}(R) \triangleq \textsf{span}(\,\cup\,_{i \, \in \, R}\{\boldsymbol{c}_i\}), \,\text{ and }\,  \boldsymbol{e}_j \notin \textsf{span}(S) \text{ for any proper subset } S \subsetneq R.
\]
In other words, $R$ is a minimal set of columns whose span in $\mathbb{F}_q$ includes $\boldsymbol{e}_j$. The minimality of $R$ implies that to recover an object, we do not use more servers than necessary. Let \( \mathcal{R}_i = \{R_{i, 1}, \hdots, R_{i, t_i}\} \) be the $t_i$ recovery sets for the object \( o_i \). Define \( \mu_l \, \in \, \mathbb{R}_{\ge 0} \) as the average rate at which the server \( l \, \in \, [n] \) processes requests for data objects, also referred to as its \emph{capacity}. The vector \( \boldsymbol{\mu} = (\mu_1, \hdots, \mu_n) \) represents the service capacities of all servers, and we assume that servers have \emph{uniform service capacity}; that is, \( \mu_j = 1, \ \forall \,\, j \in [n] \), or equivalently, \( \boldsymbol{\mu} = (\boldsymbol{1}_n)^{\top} \). Furthermore, requests for an object \( o_i \) arrive at a rate \( \lambda_i \) for all \( i \, \in \, [k] \), with the request rates for all objects represented by the vector \( \boldsymbol{\lambda} = (\lambda_1, \hdots, \lambda_k) \, \in \, \mathbb{R}_{\ge 0}^k \).

Consider the class of scheduling strategies that assign a fraction of requests for an object to each of its recovery sets. Let $\lambda_{i,j}$ be the portion of requests for object $o_i$ that are assigned to the recovery set $R_{i,j}, j \, \in \, [t_i]$. The \emph{service rate region} $\srr \subset \mathbb{R}_{\ge0}^k$ is defined as the set of all request vectors $\boldsymbol{\lambda}$ that can be served by a coded storage system generated by a matrix $\boldsymbol{G}$ and having service rate $\boldsymbol{\mu}$. Such vectors $\boldsymbol{\lambda}$ are called \textit{achievable}. Therefore, $\srr$ is the set of all vectors $\boldsymbol{\lambda}$ for which there exist $\lambda_{i,j} \, \in \, \mathbb{R}_{\ge 0}, i \, \in \, [k]$ and $j \, \in \, [t_i]$, satisfying the following constraints:
\begin{align}
        \sum_{j=1}^{t_i}\lambda_{i, j} &= \lambda_i, \quad \forall \,\, i \, \in \, [k], \label{eq:SRR_1}\\   
        \sum_{i=1}^{k}\sum_{\substack{j=1 \\ l \, \in \, R_{i, j}}}^{t_i}\lambda_{i, j} &\le \mu_l, \quad \forall \,\, l \, \in \, [n], \label{eq:SRR_2}\\
        \lambda_{i, j} &\, \in \, \mathbb{R}_{\ge 0}, \quad \forall \,\, i \, \in \, [k], j \, \in \, [t_i]. \label{eq:SRR_3}
\end{align}
Equation~\eqref{eq:SRR_1} ensures demand satisfaction for all objects, while constraints~\eqref{eq:SRR_2} ensure that no server exceeds its service capacity. The set of demand vectors $\boldsymbol{\lambda}$ for which a solution exists forms the service \textit{polytope} in $\mathbb{R}_{\ge 0}^k$. Any set $\{\lambda_{i, j}\}$ satisfying \eqref{eq:SRR_1}--\eqref{eq:SRR_3} is termed a \emph{valid allocation} for $\boldsymbol{\lambda}$; it determines the rate $\lambda_{i, j}$ at which the total request for object $i$ are routed to the $j$-th recovery set of object $i$.

In $\boldsymbol{G}_i(n,k)$, recovery sets for each basis vector $\boldsymbol{e}_j$ follow:
(i) If $i < j$, any $k$-subset of columns forms a \emph{non-systematic} recovery set for $\boldsymbol{e}_j$.
(ii) If $i \geq j$, the $j$-th column equals $\boldsymbol{e}_j$, giving a \emph{systematic} recovery set of size $1$, while any $k$-subset not containing the $j$-th column serves as a non-systematic recovery set.
\begin{remark}\label{remark:MDS}
All non-systematic recovery sets have size $k$, since no smaller set (excluding $\boldsymbol{e}_j$) can span $\boldsymbol{e}_j$ without violating the MDS property. Hence, the size of a recovery set is either $1$ or $k$.
\end{remark}


\subsection{Recovery Hypergraphs}
A \emph{hypergraph} (or simply, \emph{graph}) is a pair \((V, E)\), where \(V\) is a finite set of \emph{vertices}, and \(E\) is a multiset whose elements are subsets of \(V\), called \emph{hyperedges} (or \emph{edges}). The \emph{size} of an edge is its cardinality. A hypergraph is \emph{\(k\)-uniform} if each edge has size exactly \(k\). We further generalize uniformity to define a \((k,r)\)-quasi-uniform hypergraph: 
\begin{definition}
    A hypergraph is called \emph{\((k,r)\)-quasi-uniform} if each hyperedge has size either \(k\) or \(r\).
\end{definition}

Next, we introduce the concept of \emph{recovery hypergraph}, first defined in~\cite{Service:conf/isit/KazemiKSS20}, associated with each generator matrix. This notion enables us to reformulate our problem into a graph theory one, allowing us to leverage its well-developed tools and analytical results.

For the generator matrix \( \boldsymbol{G}_i(n,k) \), we define its associated \emph{recovery hypergraph} \( \Gamma_i(n,k) \) as follows. The hypergraph \( \Gamma_i(n,k) \) has \( n + i \) vertices:
\begin{itemize}
  \item \( n \) (column) vertices, each corresponding to a distinct column of \( \boldsymbol{G}_i(n,k) \),
  \item \( i \) auxiliary vertices, denoted \( (\boldsymbol{0}_k)_j = \bigl([0,\dots,0]^{\top}\bigr)_j \) for \( j = 1,\dots,i \). Each such vertex represents a length-\(k\) zero vector and is associated exclusively with the systematic column \( \boldsymbol{e}_j \), to which it is the only vertex connected. 
\end{itemize}
Next, we precisely describe how vertices form hyperedges in the recovery hypergraph. A set of vertices forms a hyperedge \emph{labeled} \(\boldsymbol{e}_j\) if their corresponding columns constitute a recovery set for the basis vector \(\boldsymbol{e}_j\). Because a systematic column forms a systematic recovery set by itself, introducing the auxiliary vertices $(\boldsymbol{0}_k)_j$ creates size-2 hyperedges that represent these singleton recovery sets, thereby avoiding self-loops in the hypergraph. Also, if a systematic column \(\boldsymbol{e}_j\) appears in the matrix \(\boldsymbol{G}_i\), the vertex corresponding to this systematic column is called a \emph{systematic vertex}. From the construction of the recovery graph, each systematic vertex is directly connected to an additional zero vertex \(\big([0, \dots, 0]^{\top}\big)_j\), forming a size-2 edge labeled \(\boldsymbol{e}_j\), called a \emph{systematic edge}. Edges formed by any minimal recovery set for an object $o_j$ that does not include its associated systematic vertex $\boldsymbol{e}_j$ are called \emph{non-systematic edges} (although it may contain other systematic vertices $\boldsymbol{e}_i, \, i\neq j$). Thus by Remark~\ref{remark:MDS}, every edge is either systematic (of size 2) or non-systematic (of size \(k\)). Additionally, a same vertex set may form multiple parallel edges with distinct labels if it can serve as recovery sets for multiple basis vectors. However, each hyperedge carries exactly one unique label. Figure~\ref{fig:G42-RG} illustrates the matrices \( \boldsymbol{G}_i(4,2) \) along with their corresponding recovery hypergraphs \( \Gamma_i(4,2) \) for \( i = 0,1,2 \). In these matrices, the four parity columns are explicitly given by:
\[
\bigl[\begin{smallmatrix} 1 \\ 1 \end{smallmatrix}\bigr], 
\bigl[\begin{smallmatrix} 1 \\ \alpha \end{smallmatrix}\bigr], 
\bigl[\begin{smallmatrix} 1 \\ \alpha^2 \end{smallmatrix}\bigr], 
\bigl[\begin{smallmatrix} 1 \\ \alpha^3 \end{smallmatrix}\bigr],
\]
where $\alpha$ is a primitive element of $\mathbb{F}_7$.
In \(\boldsymbol{G}_0(4,2)\), for example, vertices \(\bigl[\begin{smallmatrix} 1\\1 \end{smallmatrix}\bigr]\) and \(\bigl[\begin{smallmatrix} 1\\\alpha \end{smallmatrix}\bigr]\) form two parallel edges, one labeled \(\boldsymbol{e}_1\) and the other \(\boldsymbol{e}_2\). The two additional vertices are $\bigl[\begin{smallmatrix} 0\\0 \end{smallmatrix}\bigr]_1$ and  \(\bigl[\begin{smallmatrix} 0\\0 \end{smallmatrix}\bigr]_2\), exclusively connected to $\boldsymbol{e}_1$ and $\boldsymbol{e}_2$, respectively. In $\boldsymbol{G}_1$, for instance, edge connecting $\bigl[\begin{smallmatrix} 1\\0 \end{smallmatrix}\bigr]$ and $\bigl[\begin{smallmatrix} 0\\0 \end{smallmatrix}\bigr]_1$ is a systematic edge for $\boldsymbol{e}_1$ while edge connecting $\bigl[\begin{smallmatrix} 1\\0 \end{smallmatrix}\bigr]$ and $\bigl[\begin{smallmatrix} 1\\ \alpha^3 \end{smallmatrix}\bigr]$ is a non-systematic edge for $\boldsymbol{e}_2$, although it contains the systematic vertex $\bigl[\begin{smallmatrix} 1\\0 \end{smallmatrix}\bigr]$.

\begin{figure}[hbt]
    \centering
    \includegraphics[scale=1]{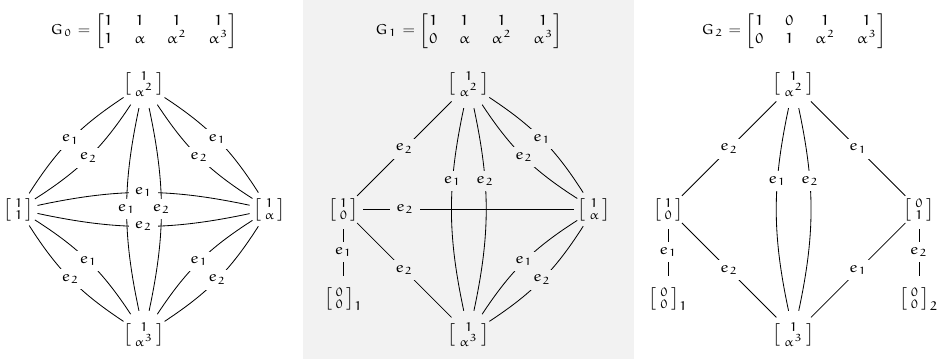}
    \caption{\(\boldsymbol{G}_i(4,2)\) matrices with their recovery hypergraphs \(\Gamma_i(4,2)\), \(i=0,1,2\).
    \label{fig:G42-RG}}
\end{figure}
Since \(\boldsymbol{G}_0\) is an MDS matrix without systematic columns, its recovery hypergraph \(\Gamma_0(n,k)\) is a \(k\)-uniform hypergraph composed exclusively of non-systematic edges of size $k$. Moreover, by Remark~\ref{remark:MDS}, for \(i = 1,\dots,k\), each systematic edge in \(\Gamma_i(n,k)\) has size 2, while all non-systematic edges have size \(k\). Hence, \(\Gamma_i(n,k)\) forms a \((k,2)\)-quasi-uniform hypergraph. Moreover, as every set of \(k\) columns is linearly independent and capable of recovering any basis vector \(\boldsymbol{e}_j\), the vertices representing these columns form \(k\) parallel non-systematic edges, each uniquely labeled by distinct \(\boldsymbol{e}_j\).
For any hypergraph \(\Gamma\) and subset \(I \subseteq [k]\), the \emph{\(I\)-induced subgraph} of \(\Gamma\) includes exactly those edges labeled by basis vectors \(\boldsymbol{e}_j\) with \(j \, \in \, I\), together with all vertices incident to these edges.
\subsection{Fractional Matching and Service Polytopes for Recovery Graphs}

A \emph{fractional matching} in a hypergraph $(V,E)$ is a vector $\boldsymbol{w} \, \in \, \mathbb{R}_{\ge 0}^{|E|}$ whose components $w_\epsilon$, for $\epsilon \, \in \, E$, are nonnegative and satisfy
\[
\sum_{\epsilon \,\ni\, v} w_\epsilon \;\le\; 1
\quad \text{for each vertex }v \, \in \, V.
\]
This constraint ensures that the total weight assigned to hyperedges incident on any given vertex does not exceed 1. In the context of our recovery hypergraph, this directly corresponds to constraint~\eqref{eq:SRR_2}, which enforces that the total request load routed through any server does not exceed its capacity (normalized to \( \mu = 1 \)). The set of all fractional matchings in $\Gamma_G=(V,E)$ forms a polytope in $\mathbb{R}_{\ge 0}^{|E|}$, called the \emph{fractional matching polytope}, denoted $\mathrm{FMP}(\Gamma_G)$. It can be written as
\[
\mathrm{FMP}(\Gamma_G) 
\;=\; 
\bigl\{\,
\boldsymbol{w} \, \in \, \mathbb{R}^{|E|} 
:\; \boldsymbol{A}\boldsymbol{w} \,\le\, \boldsymbol{1}_{|V|},\;\boldsymbol{w} \,\ge\, \boldsymbol{0}
\bigr\},
\]
where \( \boldsymbol{A} \) is the incidence matrix of \( \Gamma_G \), which is a binary matrix of size \( |V| \times |E| \), defined so that
\[
A_{v,\epsilon} = 
\begin{cases}
1, & \text{if } v \in \epsilon,\\[2pt]
0, & \text{otherwise},
\end{cases}
\]
i.e., each column of \( \boldsymbol{A} \) indicates the vertices contained in the corresponding hyperedge. Moreover, $\boldsymbol{1}$ is the all-one vector of length $|V|$, and $\boldsymbol{0}$ is the all-zero vector of length $|E|$.

\begin{example}
\label{ex:matching}
    Figure~\ref{fig:2m1s} illustrates two distinct fractional matchings on the recovery graph \( \Gamma_2(4, 2) \), associated with the generator matrix \( \boldsymbol{G}_2(4, 2) \), whose incidence matrix is
\[
\boldsymbol{A} = 
\left[
\begin{array}{cccccccc}
1 & 1 & 1 & 0 & 0 & 0 & 0 & 0 \\
1 & 0 & 0 & 0 & 0 & 0 & 0 & 0 \\
0 & 1 & 0 & 1 & 1 & 1 & 0 & 0 \\
0 & 0 & 1 & 1 & 1 & 0 & 1 & 0 \\
0 & 0 & 0 & 0 & 0 & 1 & 1 & 1 \\
0 & 0 & 0 & 0 & 0 & 0 & 0 & 1 \\
\end{array}
\right],
\]
where the six vertices and eight hyperedges are indexed from top to bottom and left to right, respectively. The two (fractional) matchings
\begin{align*}
\boldsymbol{w}_1 = [1,\ 0,\ 0,\ 0.5,\ 0,\ 0,\ 0, \ 0.75], \quad \text{and} \qquad
\boldsymbol{w}_2 = [0.5,\ 0,\ 0.5,\ 0.5,\ 0,\ 0.5,\ 0,\ 0.25]
\end{align*}
induce the same \emph{servable} request vector \( \boldsymbol{\lambda} = (1.5,\ 0.75) \). That is, for each \( j = 1,2 \), the sum of weights \( w_\epsilon \) over all hyperedges \( \epsilon \) labeled with basis vector \( \boldsymbol{e}_j \) equals \( \lambda_j \). These two matchings thus represent distinct valid allocations for the same servable request vector \( \boldsymbol{\lambda} \).
\begin{figure}[hbt]
\centering
\includegraphics[scale=0.97]{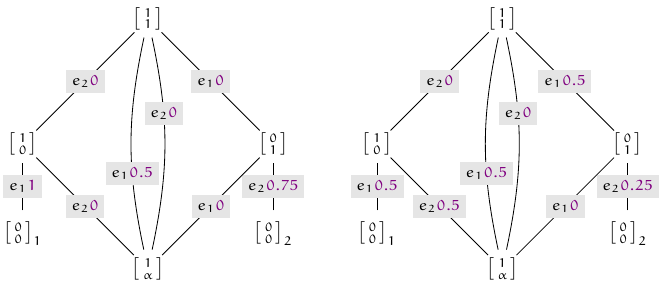}
\caption{Two different matchings in $\Gamma_2(4,2)$ having the same service vector $\boldsymbol{\lambda} = (\lambda_1,\lambda_2)=(1.5,\,0.75)$. The weight assigned to each edge is shown to its right in purple. Note that in the left figure, the weight of the edge connecting vertices \(\bigl[\begin{smallmatrix} 0\\1 \end{smallmatrix}\bigr]\) and \(\bigl[\begin{smallmatrix} 0\\0 \end{smallmatrix}\bigr]_2\) can be set up to 1.}
\label{fig:2m1s}
\end{figure}
By definition, each matching ensures that the total request assigned to each node does not exceed its capacity $\mu = 1$. Therefore, $\boldsymbol{\lambda} \, \in \, \srr$. 
\end{example}
A request vector $\boldsymbol{\lambda}$ lies within the SRR of a storage system employing a code $\boldsymbol{G}$ if and only if there exists a fractional matching $\boldsymbol{w}$ such that the sum of weights on edges labeled $\boldsymbol{e}_j$ equals $\lambda_j$, as established in the following result.
\begin{definition}
Consider a system employing an \([n, k]\) code with generator matrix \( \boldsymbol{G} \) and uniform server availability, that is, \( \boldsymbol{\mu} = \boldsymbol{1}_n \). The \emph{service rate} for \( \boldsymbol{e}_j \) under a fractional matching \( \boldsymbol{w} \), denoted \( \lambda_j(\boldsymbol{w}) \), is the sum of the weights \( w_\epsilon \) over all hyperedges \( \epsilon \) labeled by \( \boldsymbol{e}_j \). The corresponding \emph{servable request vector} is given by
\[
\boldsymbol{\lambda}(\boldsymbol{w}) = (\lambda_1(\boldsymbol{w}), \dots, \lambda_k(\boldsymbol{w})).
\]
Each \( \boldsymbol{w} \) defines a valid allocation for \( \boldsymbol{\lambda} \).
\end{definition}

\begin{proposition}[\hspace{-0.3mm}\cite{Service:journals/tit/AktasJKKS21}, Proposition 1]\label{prop:demand_to_matching}
\( \boldsymbol{\lambda} = (\lambda_1, \dots, \lambda_k) \) is achievable if and only if there exists a fractional matching \( \boldsymbol{w} \) in recovery graph \( \Gamma_G \) such that
\[
\boldsymbol{\lambda} = \boldsymbol{\lambda}(\boldsymbol{w}).
\]
\end{proposition}

Proposition~\ref{prop:demand_to_matching} shows that $\srr$ is the image of $\mathrm{FMP}(\Gamma_G)$ under a linear map from $\mathbb{R}^{|E|}$ to $\mathbb{R}^k$. Specifically, let $\boldsymbol{S}$ be the $|E|\times k$ matrix whose entries are
\[
[\boldsymbol{S}]_{\epsilon,j} 
\;=\;
\begin{cases}
1, & \text{if edge $\epsilon$ is labeled by $\boldsymbol{e}_j$,}\\[6pt]
0, & \text{otherwise}.
\end{cases}
\]
In Example~\ref{ex:matching}, the matrix $\boldsymbol{S} \in \{0, \, 1\}^{8\times 2}$ associated with $\Gamma_2$ is given by
\[
\boldsymbol{S} = \left[
\begin{array}{cccccccc}
1 & 0 & 0 & 1 & 0 & 1 & 1 & 0 \\
0 & 1 & 1 & 0 & 1 & 0 & 0 & 1 \\
\end{array}
\right]^{\top},
\]
Since each edge is associated with exactly one label, each row of $\boldsymbol{S}$ contains exactly one entry equal to 1, with all other entries equal to 0. For a matching $\boldsymbol{w}\, \in \, \mathrm{FMP}(\Gamma_G)$, the resulting service vector is $\boldsymbol{\lambda}(\boldsymbol{w}) = \boldsymbol{w}\,\boldsymbol{S}$. If $\boldsymbol{\lambda} \, \in \, \srr$, there may be multiple $\boldsymbol{w}\, \in \, \mathrm{FMP}(\Gamma_G)$ such that $\boldsymbol{\lambda} = \boldsymbol{w}\,\boldsymbol{S}$.

The service rate region $\srr$, therefore, also forms a polytope in $\mathbb{R}_{\ge 0}^k$, and we use the terms \emph{service polytope} and \emph{service rate region} interchangeably. We denote the SRR corresponding to $\boldsymbol{G}_i(n,k)$ by $\mathcal{S}_i(n,k)$ or simply $\mathcal{S}_i$. 
\begin{remark}\label{rm:feasibility_check}
From the above argument, Proposition~\ref{prop:demand_to_matching} can also be stated as follows:
\[
\boldsymbol{\lambda} = (\lambda_1, \dots, \lambda_k) \text{ is achievable if and only if the following linear program is feasible:}
\]
\begin{align}\label{eq:SRR_reformulation}
\boldsymbol{\lambda} & = \boldsymbol{w}\boldsymbol{S},\\
\text{s.t. }\quad \boldsymbol{A}\boldsymbol{w} & \le \boldsymbol{1}_n, \notag \\
\boldsymbol{w} & \ge \boldsymbol{0}_n.\notag
\end{align}
When $\boldsymbol{\lambda}$ is achievable, the problem can generally have infinitely many solutions. Such solutions lie in a space whose dimension is $|E| - \mathrm{rank}(\boldsymbol{S}) \ge |E| - k$.
\end{remark}
The \emph{size} of a matching $\boldsymbol{w}$ is $\sum\limits_{\epsilon \, \in \, E} w_{\epsilon}$. The \emph{(fractional) matching number} $\nu^*(\Gamma_G)$ is the maximum possible matching size:
\[
\nu^*(\Gamma_G) \;=\; \max_{\boldsymbol{w} \, \in \, \mathrm{FMP}(\Gamma_G)}\sum_{\epsilon \,\, \in \,\, E} w_{\epsilon}.
\]
Intuitively, the size of a matching quantifies the total weight distributed across all hyperedges, and the matching number captures the largest such total achievable under vertex capacity constraints. For instance, the size of two matchings given in Example~\ref{ex:matching} are both $1.5 + 0.75 = 2.25$. A \emph{fractional vertex cover} of $(V,E)$ is a vector $\boldsymbol{d} \, \in \, \mathbb{R}^{|V|}$ with nonnegative components $d_v$ such that $\sum_{v \, \in \, \epsilon} d_v \ge 1$ for every edge $\epsilon \, \in \, E$. Its \emph{size} is $\sum\limits_{v \, \in \, V} d_v$. The \emph{(fractional) vertex cover number} $\tau^*(\Gamma_G)$ is the minimum size of any fractional vertex cover:
\[
\tau^*(\Gamma_G)
\;=\;
\min_{\boldsymbol{d} \,\ge\, \boldsymbol{0}}
\Bigl\{\,
\sum_{v \, \in \, V} d_v 
:\; \boldsymbol{A}^{\top}\boldsymbol{d} \,\ge\, \boldsymbol{1}_{|E|}
\Bigr\}.
\]
Intuitively, a fractional vertex cover assigns weights to vertices such that each hyperedge receives a total weight of at least one. Its size reflects the total vertex weight required to meet this condition, and the vertex cover number denotes the minimum such total over all valid covers. Finding $\nu^*(\Gamma_G)$ is a linear program whose dual problem finds the minimum fractional vertex cover $\tau^*(\Gamma_G)$. By the theory of linear programming duality~\cite{Linear_programming}, we establish the following properties (see Appendix~\ref{app:duality} for a detailed derivation):
\begin{itemize}
    \item \textbf{Weak Duality:} The size of \emph{any} fractional vertex cover is an upper bound on the size of \emph{any} fractional matching.
    \item \textbf{Strong Duality:} The optimal values coincide, i.e., $\nu^*(\Gamma_{G}) = \tau^*(\Gamma_{G})$.
\end{itemize}


Framing the SRR problem within graph theory allows us to leverage known results from the literature. The former is seen in Proposition~\ref{prop:demand_to_matching}, while the latter follows from the next result, which can easily be proved using Proposition~\ref{prop:demand_to_matching}.

\begin{proposition}\label{prop:sum_bound}
For any vector $\boldsymbol{\lambda} = (\lambda_1, \lambda_2, \dots, \lambda_k)$ in the service region $\srr$,
\begin{equation}\label{eq:Duality_bound}
\tau^*(\Gamma_G) \;\ge\; \sum_{i=1}^k \lambda_i.
\end{equation}
\end{proposition}

\begin{proof}
Since $\boldsymbol{\lambda} \, \in \, \srr$, by Proposition~\ref{prop:demand_to_matching} there exists a matching $\boldsymbol{w}$ such that $\boldsymbol{\lambda} = \boldsymbol{\lambda}(\boldsymbol{w}) = \boldsymbol{w}\boldsymbol{S}$. Denote by $\boldsymbol{1}_k$ the length-$k$ column vectors of all ones. Then,
\begin{align}
\sum_{j=1}^k \lambda_j(\boldsymbol{w})
\; =\; \boldsymbol{\lambda}(\boldsymbol{w}) \cdot \boldsymbol{1}_k
\;=\; \boldsymbol{w}(\boldsymbol{S} \cdot \boldsymbol{1}_k)
\;=\; \boldsymbol{w} \cdot \boldsymbol{1}_k
\;=\; \sum_{\epsilon \, \in \, E} w_{\epsilon}
\; \le\;
\max_{\boldsymbol{w} \, \in \, \mathrm{FMP}(\Gamma_G)} \sum_{\epsilon \, \in \, E} w_{\epsilon}
\;=\;
\nu^*(\Gamma_G)
\;=\;
\tau^*(\Gamma_G).
\end{align}
(The third equality follows from the fact that each row of $\boldsymbol{S}$ has exactly one entry equal to 1, with all others equal to 0. The second-to-last equality follows from the definition of $\nu^*(\Gamma_G)$, and the last equality comes from the LP duality.)
\end{proof}
\begin{example}
\label{ex:vc_bound}
In Example~\ref{ex:matching}, assigning weight \( c = 1 \) to each of the three vertices 
\( \bigl[\begin{smallmatrix} 1\\0 \end{smallmatrix}\bigr] \), 
\( \bigl[\begin{smallmatrix} 0\\1 \end{smallmatrix}\bigr] \), and 
\( \bigl[\begin{smallmatrix} 1\\1 \end{smallmatrix}\bigr] \), 
and setting \( c = 0 \) for all remaining vertices yields a vertex cover of size 3. This implies that for any service vector \( \boldsymbol{\lambda} = (\lambda_1, \lambda_2) \in \srr \), the following inequality must hold:
\[
3 \ge \tau^*(\Gamma_2) \ge \lambda_1 + \lambda_2.
\]
It is also easy to verify that this is actually the smallest vertex cover, i.e., $3 = \tau^*(\Gamma_2) = \nu^*(\Gamma_2)$.
\end{example}

\begin{remark}
This result generalizes and tightens Proposition 2 in~\cite{Service:journals/tit/AktasJKKS21}, where the authors proved that the size of any integral vertex cover is an upper bound on the achievable heterogeneous sum rate. It establishes a fundamental bound on the sum of request rates within the service polytope of any linearly coded storage system. Consequently, the size of any (fractional) vertex cover serves as an upper bound on the sum rate for any achievable vector \( \boldsymbol{\lambda} \). In particular, for large hypergraphs with specific structures, such as when the number of vertices \( n \) is significantly smaller than the number of edges \( \binom{n}{k} \) or when most edges have a large cardinality, this lemma provides a simple yet tight estimate of the achievable sum rate.
\end{remark}

\begin{remark}\label{rm:sum_bound_subgraph}
A straightforward manipulation of Proposition~\ref{prop:sum_bound} shows that it also applies to any subgraph of $\Gamma_G$. Specifically, if $I$ is any subset of $[k]$ and $\Gamma'$ is the $I$-induced subgraph of $\Gamma$, then
\[
\tau^*(\Gamma') \;\ge\; \sum_{i \, \in \, I} \lambda_i.
\]
\end{remark}
\subsection*{Main Notations Summary}
\begin{itemize}
\setlength{\itemindent}{.25in}
\setlength{\itemsep}{1pt}
  \setlength{\parskip}{0pt}
  \setlength{\parsep}{0pt}
  \item[$k$ --] number of data objects.
  \item[$n$ --] number of servers (storage nodes).
  \item[$i$ --] number of systematic columns (or systematic nodes) in the MDS generator matrix.
  \item[$\lambda_j$ --] the rate of requests for the $j$-th object.
  \item[$\boldsymbol{\lambda}$ --] vector of request rates $\boldsymbol{\lambda} = (\lambda_1,\lambda_2,\dots ,\lambda_n)$
  \item[$\boldsymbol{G}$ --] Generator matrix.
  \item [$\mathcal{S}(\boldsymbol{G})$ --] Service rate region (SRR) of the system using code $\boldsymbol{G}$.
  \item [$R$ --] A recovery set for a data object (basis vector) in the code.
  \item [$\Gamma_G$ --] Recovery hypergraph associated with the generator matrix $\boldsymbol{G}$.
  \item [$\boldsymbol{w}$, $\boldsymbol{d}$ --] Fractional matching weight vector and fractional vertex cover weight vector, respectively.
  \item [$\nu^*(\Gamma_G)$ --] Fractional matching number the hypergraph $\Gamma_G$.
  \item [$\tau^*(\Gamma_G)$ --] Fractional vertex cover number of the hypergraph $\Gamma_G$. Duality theorem gives $\nu^*(\Gamma_G) = \tau^*(\Gamma_G)$.
  \item [$\mathbb{R}^k_+$ --] positive orthant of $\mathbb{R}^k$.
\end{itemize}
\section{Two Bounding Simplices}\label{Sec:Bounding_simplices}
The simplices of interest in \(\mathbb{R}^k\) are convex hulls of the origin and the \(k\) axis-aligned points \(c_j\boldsymbol{e}_j\), for \(j = 1, \dots, k\), where each \(c_j > 0\). When all \(c_j\) are equal to a common constant \(c \ge 0\), the resulting simplex corresponds to the set of points \((x_1, x_2, \dots, x_k) \, \in \, \mathbb{R}^k_+\) satisfying $\sum_{j=1}^{k} x_j \le c.$

This section introduces two simplices related to any SRR: the \emph{Maximal matching simplex} and the \emph{Maximal achievable simplex}. The former is the smallest simplex that contains the SRR, while the latter is the largest simplex fully contained by it. These simplices generally help localize the SRR.
In Section~\ref{Sec:Inclusion}, we will use them to derive an inclusion theorem for \(\mathcal{S}_i(n,k)\) that shows the role of systematic nodes.

\subsection{Maximal Matching Simplex}
For each \(\mathcal{S}_i(n,k)\), we seek an outer simplex containing the service rate region, defined as 
\[
\sum_{j=1}^{k}\lambda_j \;\le\; c, ~~ \lambda_j\ge 0.
\]
where \(c\) is a constant specific to \(\mathcal{S}_i(n,k)\). Proposition~\ref{prop:sum_bound} provides exactly this bound. Recall from Eq.~\eqref{eq:Duality_bound} that for any generator matrix \(\boldsymbol{G}_i(n,k)\) and any service vector \(\boldsymbol{\lambda} \, \in \, \mathcal{S}_i(n,k)\), the SRR lies within the positive orthant simplex defined by
\[
\sum_{j=1}^{k}\lambda_j \;\le\; \tau^*(\Gamma_i(n,k) = \nu^*(\Gamma_i(n,k)).
\]
We refer to this region as the \emph{maximal matching simplex}. We next determine the value of \(\tau^*(\Gamma_i(n,k))\) by constructing explicit matchings and vertex covers in two distinct cases.

\paragraph{Case 1: \(n - i \ge k\)}

Consider the following vertex cover of \(\Gamma_i(n,k)\):
\[
d_v \;=\;
\begin{cases}
	1, & \text{if \(v\) is a systematic vertex}, \\[4pt]
	{1}/{k}, & \text{otherwise}.
\end{cases}
\]
We show that this is a valid vertex cover:
\begin{itemize}
\item For a \emph{systematic} edge \(\epsilon\), there is exactly one vertex \(v\) with \(d_v = 1\).
\item For a \emph{non-systematic} edge \(\epsilon\), Lemma~\ref{remark:MDS} guarantees that \(\epsilon\) has size \(k\). Hence, \(\sum\limits_{v \, \in \, \epsilon} d_v = k \cdot \tfrac{1}{k} = 1\).
\end{itemize}
Thus, the total weight assigned to the vertices incident to each edge is at least 1, so this indeed forms a vertex cover of size \(i + \dfrac{n-i}{k}\). On the other hand, we construct a matching by assigning:
\[
w_{\epsilon} \;=\;
\begin{cases}
	1, & \text{if edge \(\epsilon\) is systematic},\\[3pt]
	0, & \text{if \(\epsilon\) is non-systematic and contains a systematic vertex},\\[3pt]
	\bigl(k\,\binom{n-i-1}{k-1}\bigr)^{-1}, & \text{otherwise}.
\end{cases}
\]
We check that this matching is valid as follows:
\begin{itemize}
\item A systematic vertex \(v\) is incident to exactly one systematic edge with weight 1.
\item A non-systematic vertex \(v\) belongs to \(\binom{n-i-1}{k-1}\) different size-\(k\) sets of columns that do not include any systematic column. Each such set corresponds to \(k\) parallel edges, each assigned weight \(\bigl(k\,\binom{n-i-1}{k-1}\bigr)^{-1}\). Summing over all these edges gives
\[
\sum_{\epsilon \ni v} w_{\epsilon}
\;=\;
k \,\binom{n-i-1}{k-1}
\,\times\,
\left(k\,\binom{n-i-1}{k-1}\right)^{-1}
\;=\; 1.
\]
\end{itemize}
Hence, the total weight assigned to the edges incident to each vertex $v$ does not exceed 1. Counting the edges:
\begin{itemize}
\item There are \(i\) systematic vertices, each of which is incident to one systematic edge.
\item Excluding these, the remaining $n-i$ non-systematic vertices form a \(k\)-uniform hypergraph. Because \(n-i \ge k\), there are \(k \,\binom{n-i}{k}\) edges of size $k$ in this graph.
\end{itemize}
The total matching size is therefore
\[
i 
\;+\; 
\left(k\,\binom{n-i-1}{k-1}\right)^{-1}
\,\cdot\, 
\left(k\,\binom{n-i}{k}\right)
\;=\;
i + \frac{n-i}{k}.
\]
By LP duality, since there is a vertex cover and a matching with the same size, 
\begin{equation}\label{eq:matching_number}
\nu^*(\Gamma_i(n, k)) 
\;=\; 
\tau^*(\Gamma_i(n, k))  
\;=\; 
i + \frac{n-i}{k}.
\end{equation}
Thus, in this case, the Maximal matching simplex is
\begin{align}
\begin{cases}\label{eq:maximal_matching_1}
	\boldsymbol{0} \le \boldsymbol{\lambda}, \\[5pt]
	\sum_{j=1}^{k}\lambda_j \;\le\; i + \dfrac{n-i}{k}.
\end{cases}
\end{align}
Note that when \(i=0\), \(\Gamma_0(n,k)\) reduces to the \(k\)-uniform hypergraph of $n$ vertices, where 
\begin{align}\label{eq:matching_non_systematic}
\nu^*(\Gamma_0(n, k)) = \tau^*(\Gamma_0(n, k)) = {n}/{k}.
\end{align}
\paragraph{Case 2: \(n - i < k\)}

In this regime, consider the following vertex cover in \(\Gamma_i(n,k)\):
\[
d_v \;=\;
\begin{cases}
	1, & \text{if \(v\) is systematic},\\
	0, & \text{otherwise}.
\end{cases}
\]
Because \(n - i < k\), any non-systematic edge has size \(k\) and must include at least one systematic vertex, so all edges are covered. The total weight is \(i\). Next, define a matching:
\[
w_{\epsilon} \;=\;
\begin{cases}
	1, & \text{if \(\epsilon\) is systematic},\\
	0, & \text{otherwise}.
\end{cases}
\]
Since the systematic edges do not overlap on any vertex, they form a valid matching of size \(i\). Therefore, by duality,
\begin{equation}\label{eq:matching_number_2}
\nu^*(\Gamma_i(n, k)) 
\;=\; 
\tau^*(\Gamma_i(n, k))  
\;=\; 
i.
\end{equation}
Hence, in this case, the Maximal matching simplex is
\begin{align}
\begin{cases}\label{eq:maximal_matching_2}
	\boldsymbol{0} \le \boldsymbol{\lambda}, \\[3pt]
	\sum_{j=1}^{k}\lambda_j \;\le\; i.
\end{cases}
\end{align}

\subsection{Axes-Intercept Points and the Maximal Achievable Simplex}

We now characterize a simplex contained in \(\srr\). The \emph{axes-intercept} vertices of this simplex coincide with those of \(\srr\). For each \(j \, \in \, [k]\), define
\[ 
\lambda_j^{\text{int}} \triangleq \max_{ \gamma\cdot\boldsymbol{e}_j \, \in \, \srr} \gamma. \] 
That is, \(\lambda_j^{\text{int}}\) is the maximum achievable demand for the
object \(j\) when all other demands are zero. 
For example, when $j=1$, we have \(\lambda_1^{\text{int}} = \max\{\, \gamma \mid (\gamma, 0, \dots, 0) \, \in \, \srr \}\). We claim that \(\lambda_j^{\text{int}} = \max\limits_{\boldsymbol{\lambda} \, \in \, \srr} \lambda_j \triangleq \lambda_j^{\max}\). To see that, consider $j=1$, wlog. Suppose, for contradiction, that there exists
\(\boldsymbol{\eta} = (\eta_1, \eta_2, \dots, \eta_k) \, \in \, \srr\) with \(\eta_1 > \lambda_1^{\text{int}}\). Consider the vector \(\boldsymbol{\eta}' = (\eta_1, 0, \dots, 0)\). Since \(\boldsymbol{\eta}\) satisfies constraints Eq.~\eqref{eq:SRR_1}–Eq.~\eqref{eq:SRR_3}, and these constraints remain valid when all coordinates but one are set to zero, it follows that \(\boldsymbol{\eta}' \, \in \, \srr\), contradicting the definition of \(\lambda_1^{\text{int}}\). Therefore, \(\lambda_1^{\text{int}} = \lambda_1^{\max}\), and similarly \(\lambda_j^{\text{int}} = \lambda_j^{\max}\) for all \(j \, \in \, [k]\).



Geometrically, \( \lambda_j^{\text{max}} \) represents the intercept of the service polytope with the axis defined by \( \boldsymbol{e}_j \). Practically, $\lambda_j^{\text{max}}$ quantifies the highest individual request rate for object \( o_j \) that the system can support in isolation, i.e., when requests for all other objects are set to zero. Characterizing these values is particularly relevant in distributed storage systems, where object demands or popularities are commonly skewed~\cite{Service:journals/tit/AktasJKKS21}. We therefore call it the maximum achievable demand for $\lambda_j$. Define the simplex \( \mathcal{A} \) as:
\[
\mathcal{A} = \textsf{conv}\Bigl(\bigl\{\boldsymbol{0}_k, \lambda_1^{\text{max}}\boldsymbol{e}_1, \lambda_2^{\text{max}}\boldsymbol{e}_2, \dots, \lambda_k^{\text{max}}\boldsymbol{e}_k\bigr\}\Bigr),
\]
where \( \textsf{conv}(\mathcal{T}) \) denotes the convex hull of a set \( \mathcal{T} \), defined as \( \mathcal{T} = \{\boldsymbol{v}_1, \dots, \boldsymbol{v}_p\} \subset \mathbb{R}^k \). Specifically, \( \textsf{conv}(\mathcal{T}) \) consists of all convex combinations of the elements in \( \mathcal{T} \), that is, all vectors of the form
\[
\sum_{i=1}^p \gamma_i \boldsymbol{v}_i, \quad \text{where } \gamma_i \ge 0 \text{ and } \sum_{i=1}^p \gamma_i = 1,
\]
as described in~\cite{ConvexAnalysis:books/Rockafellar70}. In \cite[Lemma 1]{Service:conf/isit/KazemiKS20} it was shown that
$\srr$ is a non-empty, convex, closed, and bounded subset of $\mathbb{R}_{\ge 0}^k$. Thus \( \mathcal{A} \subseteq \mathcal{S}(\boldsymbol{G}) \). This implies that all points within the simplex \( \mathcal{A} \) are achievable, and we refer to it as the \textit{Maximal achievable simplex}. Therefore, characterizing these extreme points is of significant interest. The following theorem helps us to achieve this.

\begin{proposition}\label{prop:MaxAchievableSimplex}
For \(\mathcal{S}_i(n,k)\),
\[
\lambda_j^{\text{max}} 
\;=\;
\begin{cases}
	1 + \dfrac{n-1}{k}, & \text{if } j \le i \text{ and } n-1 \ge k,\\[6pt]
	1, & \text{if } j \le i \text{ and } n = k,\\[6pt]
	\dfrac{n}{k}, & \text{if } j > i.
\end{cases}
\]
\end{proposition}

\begin{proof}
Consider the axis-intercept vector \(\lambda_j^{\text{max}}\boldsymbol{e}_j\), which has all components equal to zero except \(\lambda_j\). Since all other rates are zero, in \(\Gamma_i(n,k)\), we remove all edges with labels different from $\boldsymbol{e}_j$. Let \(\Gamma_i^j(n,k)\) be the subgraph of \(\Gamma_i(n,k)\) that contains only edges labeled \(\boldsymbol{e}_j\). Then by Remark~\ref{rm:sum_bound_subgraph},
\[
\lambda_j^{\text{max}}
\;=\;
\nu^*\bigl(\Gamma_i^j(n,k)\bigr).
\]

\emph{Case 1: \(j > i\).}
All edges labeled \(\boldsymbol{e}_j\) are non-systematic, so \(\Gamma_i^j(n,k)\) is a \(k\)-uniform hypergraph on $n$ vertices. By Eq.~\eqref{eq:matching_non_systematic},
\(
\lambda_j^{\text{max}}
\;=\;
{n}/{k}.
\)

\emph{Case 2: \(j \le i\).}
For \(\boldsymbol{e}_j\), the subgraph \(\Gamma_i^j(n,k)\) contains:
\begin{itemize}
  \item One systematic edge of size 1 (the column \(\boldsymbol{e}_j\) itself).
  \item Non-systematic edges of size \(k\) formed by choosing any \(k\)-subset of columns excluding \(\boldsymbol{e}_j\).
\end{itemize}
Hence, \(\Gamma_i^j(n,k)\) splits into two disjoint hypergraphs: a 2-uniform hypergraph with 2 nodes (the systematic column and its associated zero vector) and a \(k\)-uniform hypergraph with \((n-1)\) nodes (all remaining columns). Therefore by Eq.~\eqref{eq:matching_number} and Eq.~\eqref{eq:matching_number_2},
\[
\nu^*(\Gamma_i^j) 
\;=\; 
\begin{cases}
1, & \text{if } n = k,\\
1 + (n-1)/k, & \text{if } n-1 \ge k,
\end{cases}
\]
matching the desired piecewise form.
\end{proof}
\begin{example}
\label{ex:srr}
    Figure~\ref{fig:G42-SRR} (a) illustrates the service region \(\mathcal{S}_2(4,2)\) along with its two bounding simplices. Some points on the Maximal Matching Simplex lie within \(\mathcal{S}_2(4,2)\), while others lie outside, reflecting that the true service region is strictly between these two simplices. See also Example~\ref{ex:vc_bound}.
\end{example}
From the derived expressions of these two simplices, we see that the Maximal Matching Simplex is always within a factor of $k$ of the Maximal Achievable Simplex. The worst case (when the two are furthest apart) occurs with $\mathcal{S}_k(k, k)$ (systematic encoding of $k$ data objects into $k$ servers), where the Maximal Matching Simplex is given by $\sum_{j=1}^k \lambda_j \le k$ and the Maximal Achievable Simplex is given by $\sum_{j=1}^k \lambda_j \le 1$, which is $k$-times smaller. Figure~\ref{fig:G42-SRR} (b) depicts this scenario with $\mathcal{S}_2(2,2)$.

\begin{figure}[ht]
    \centering    
    \subfloat[]{%
    \resizebox{0.42\textwidth}{!}{
        \begin{tikzpicture}[scale=1]
            \draw[thick,->] (0,0) -- (4.8,0) node[below] {$\lambda_b$};
            \draw[thick,->] (0,0) -- (0,4.8) node[left] {$\lambda_a$};
            \draw[dotted,thick] (0,3.33) -- (3.33,0);
            \draw[dotted,thick] (0,4) -- (4,0);
            \draw[thick,red] (0,3.33) -- (1.33,2.66) -- (2.66,1.33) -- (3.33,0);
            \node at (0,3.33) [left] {\textcolor{red}{2.5}};
            \node at (3.33,0) [below] {\textcolor{red}{2.5}};
            \node at (0,4) [left] {\textcolor{red}{3}};
            \node at (4,0) [below] {\textcolor{red}{3}};
            \draw[-] (2.93,0.4) -- (4, 1) node[right, fill=gray!20, inner sep=2pt] {maximal achievable simplex};
            \draw[-] (0.7,3.3) -- (2.5,4) node[right, fill=gray!20, inner sep=2pt] {maximal matching simplex};
            \draw[-] (2.9,0.85) -- (3.2,2.3) node[right, fill=gray!20, inner sep=2pt] {service polytope};
        \end{tikzpicture}%
    }
    }
    \hspace{0.5cm}
    \subfloat[]{%
       \resizebox{0.42\textwidth}{!}{
       \begin{tikzpicture}[scale=1]
            \draw[thick,->] (0,0) -- (4.7,0) node[below] {$\lambda_b$};
            \draw[thick,->] (0,0) -- (0,4.7) node[left] {$\lambda_a$};
            \draw[dotted,thick] (0,2) -- (2,0);
            \draw[dotted,thick] (0,4) -- (4,0);
            \draw[thick,red] (0,2) -- (2,2) -- (2,0);
            \node at (0,2) [left] {\textcolor{red}{1}};
            \node at (2,0) [below] {\textcolor{red}{1}};
            \node at (0,4) [left] {\textcolor{red}{2}};
            \node at (4,0) [below] {\textcolor{red}{2}};
            \draw[-] (1.3,0.7) -- (4, 1) node[right, fill=gray!20, inner sep=2pt] {maximal achievable simplex};
            \draw[-] (1.2,2.8) -- (2.5,4) node[right, fill=gray!20, inner sep=2pt] {maximal matching simplex};
            \draw[-] (2,1.5) -- (3.2,2.3) node[right, fill=gray!20, inner sep=2pt] {service polytope};
        \end{tikzpicture}%
    }
    }
    \caption{(Left) The service region \(\mathcal{S}_2(4,2)\) (in red) together with its Maximal matching and Maximal achievable simplices, where $\lambda^{\text{int}}_1 = \lambda^{\text{int}}_2 = 2.5$. Some points on the Maximal matching simplex are indeed achievable. (Right) $\mathcal{S}_2(2, 2)$ and its two bounding simplices, where $\lambda^{\text{int}}_1 = \lambda^{\text{int}}_2 = 1$. In this case, these two bounding simplices are ``further apart" from each other than in case (a).}
    \label{fig:G42-SRR}
\end{figure}
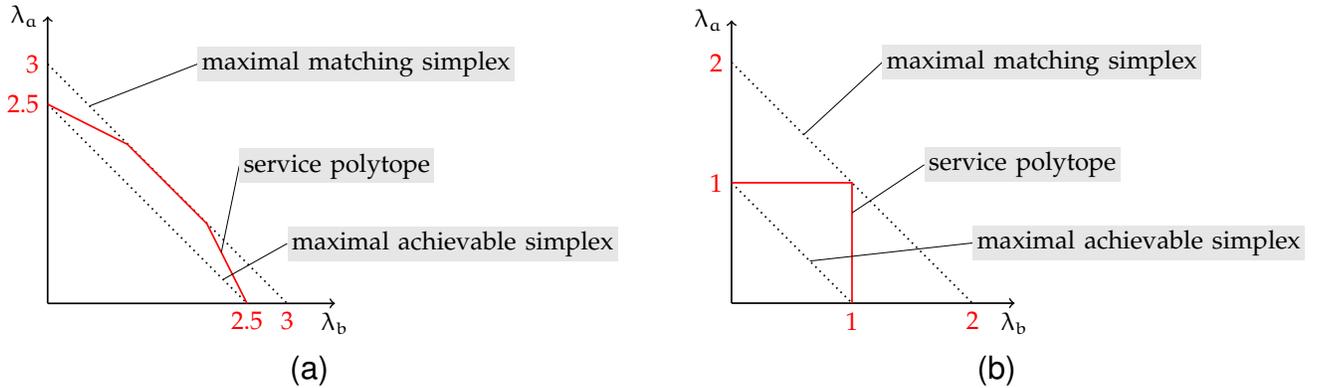

We observe that for MDS codes, the two bounding simplices of the SRR can differ by a factor proportional to the number of objects $k$. In contrast, for Reed–Muller codes, the bounds differ by a constant factor of at most 2, independent of the code parameters~\cite{SRR_RM:conf/isit,Service:preprint/arxiv/LySL25}.

When \(i = 0\), \(\Gamma_0(n,k)\) is a \(k\)-uniform hypergraph, so these two simplices coincide exactly with the SRR polytope. Concretely, (cf.~\cite{Service:journals/tit/AktasJKKS21})
\[
 \mathcal{S}_0(n,k) = \begin{cases}
	\boldsymbol{0} \le \boldsymbol{\lambda},\\
	\sum_{j=1}^k \lambda_j \;\le\; {n}/{k}.
\end{cases}
\]
Moreover, from Eq.~\eqref{eq:maximal_matching_1}, Eq.~\eqref{eq:maximal_matching_2} and Proposition~\ref{prop:MaxAchievableSimplex}, we observe that for fixed $n$ and $k$, decreasing $i$ reduces the gap between the two bounding simplices, enabling a more precise characterization of the service polytope. Similarly, for fixed $i$ and $k$, increasing $n$ produces the same effect, improving the accuracy with which the SRR can be located.

\section{An Inclusion Theorem for \(\mathcal{S}_i(n, k)\), \(0 \le i \le k\)}\label{Sec:Inclusion}
In this section, we show that increasing the number of systematic nodes (or servers) strictly enlarges the service rate region, offering meaningful guidance for practical system design. While the formal proof relies on graph-theoretic arguments, the underlying intuition is straightforward: a coded server can only participate in data recovery by forming recovery sets with \(k-1\) other servers, whereas a systematic (uncoded) server can serve its associated data object independently by itself. As a result, systematic nodes reduce the overall capacity required for recovery. The result is formalized in the following theorem.
\begin{theorem}\label{theo:polytopeInclusion}
For each \(k \ge 2\) and \(n \ge k\)
\begin{equation}
\mathcal{S}_0 \;\subsetneq\; \mathcal{S}_1 \;\subsetneq\; \mathcal{S}_2 
\;\subsetneq\; \dots \;\subsetneq\; \mathcal{S}_{k-1} \;\subsetneq\; \mathcal{S}_k.
\label{eq:inclusion}
\end{equation}
In other words, having more systematic columns strictly enlarges the SRR polytopes.
\end{theorem}

\begin{proof}
From the definition of MDS matrices, any set of \(k\) columns excluding $\boldsymbol{e}_j$ in \(\boldsymbol{G}_i(n,k)\) can serve as a recovery set for each basis vector \(\boldsymbol{e}_j\). In the corresponding recovery hypergraph \(\Gamma_i(n,k)\), the number of hyperedges labeled with \(\boldsymbol{e}_j\) is
\[
\begin{cases}
1 + \binom{n-1}{k}, & \text{if \(\boldsymbol{G}_i(n,k)\) contains the systematic column \(\boldsymbol{e}_j\), i.e., if \(i \ge j\)},\\
\binom{n}{k}, & \text{if \(\boldsymbol{G}_i(n,k)\) does not contain \(\boldsymbol{e}_j\), i.e.\ \(i < j\)}.
\end{cases}
\]
Additionally, in each recovery hypergraph \(\Gamma_G\), any systematic recovery set is dedicated to exactly one basis vector (i.e., one object), whereas $k$-column recovery sets can recover any basis vector.

\noindent
\textbf{Comparing \(\boldsymbol{G}_i\) and \(\boldsymbol{G}_{i+1}\).}
By construction in Section~\ref{Sec:Problem_Formulate}, \(\boldsymbol{G}_i\) and \(\boldsymbol{G}_{i+1}\) differ in exactly one column: \(\boldsymbol{G}_{i+1}\) is obtained by replacing the \((i+1)\)-th column in \(\boldsymbol{G}_i\) (i.e., parity column $p_{i+1}$) with the systematic column \(\boldsymbol{e}_{i+1}\). All other columns remain the same, and any systematic column in \(\boldsymbol{G}_i\) still appears in \(\boldsymbol{G}_{i+1}\).

\noindent
\textbf{Case Analysis on \(\boldsymbol{e}_j\).}
\begin{itemize}
    \item \textbf{Case 1:} \(j > i\). Then there is no systematic recovery set for \(\boldsymbol{e}_j\) in \(\boldsymbol{G}_i\). There might be a systematic recovery set for \(\boldsymbol{e}_j\) in \(\boldsymbol{G}_{i+1}\), specifically if \(j = i+1\). In \(\boldsymbol{G}_i\), recovery for \(\boldsymbol{e}_j\) is possible only via non-systematic edges of size \(k\). Whenever such a recovery set contains the replaced parity column $p_{i+1}$, we can form a corresponding recovery set in \(\boldsymbol{G}_{i+1}\) by substituting \(\boldsymbol{e}_{i+1}\) for $p_{i+1}$. Because \(\boldsymbol{G}_{i+1}\) is an MDS matrix, these new \(k\)-column sets are linearly independent and also recover \(\boldsymbol{e}_j\). Hence, any service rate achievable in \(\mathcal{S}_i\) is still feasible in \(\mathcal{S}_{i+1}\).
    \item \textbf{Case 2:} \(j \le i\). In this scenario, \(\boldsymbol{e}_j\) is already a systematic column in \(\boldsymbol{G}_i\), so it remains in \(\boldsymbol{G}_{i+1}\). Proceeding similarly as in Case 1, all service configurations in \(\mathcal{S}_i\) remain feasible in \(\mathcal{S}_{i+1}\).
\end{itemize}
Thus, \(\mathcal{S}_i \subseteq \mathcal{S}_{i+1}\).

\medskip
\noindent
\textbf{Strict Inclusion.}
To see that \(\mathcal{S}_i\) is \emph{strictly} contained in \(\mathcal{S}_{i+1}\), consider the new systematic column \(\boldsymbol{e}_{i+1}\), which is present in \(\boldsymbol{G}_{i+1}\) but not in \(\boldsymbol{G}_i\). From Proposition~\ref{prop:MaxAchievableSimplex}, the maximum single-basis service rate for \(\boldsymbol{e}_{i+1}\) when it is systematic is
\[
\lambda_{i+1}^{\max}(\boldsymbol{G}_{i+1})
\;=\;
1 + \frac{n-1}{k}
\;>\;
\frac{n}{k}
\;=\;
\lambda_{i+1}^{\max}(\boldsymbol{G}_i),\,\text{for all } k \ge 2.
\]
Hence, the single-axis service vector 
\(
(0,\dots,0,\,1+(n-1)/{k},\,0,\dots,0)
\)
is in \(\mathcal{S}_{i+1}\) but not in \(\mathcal{S}_i\). Therefore, 
\(\mathcal{S}_i \subsetneq \mathcal{S}_{i+1}\).

\end{proof}

\begin{figure}[hbt]
    \centering
    \includegraphics[scale=0.33]{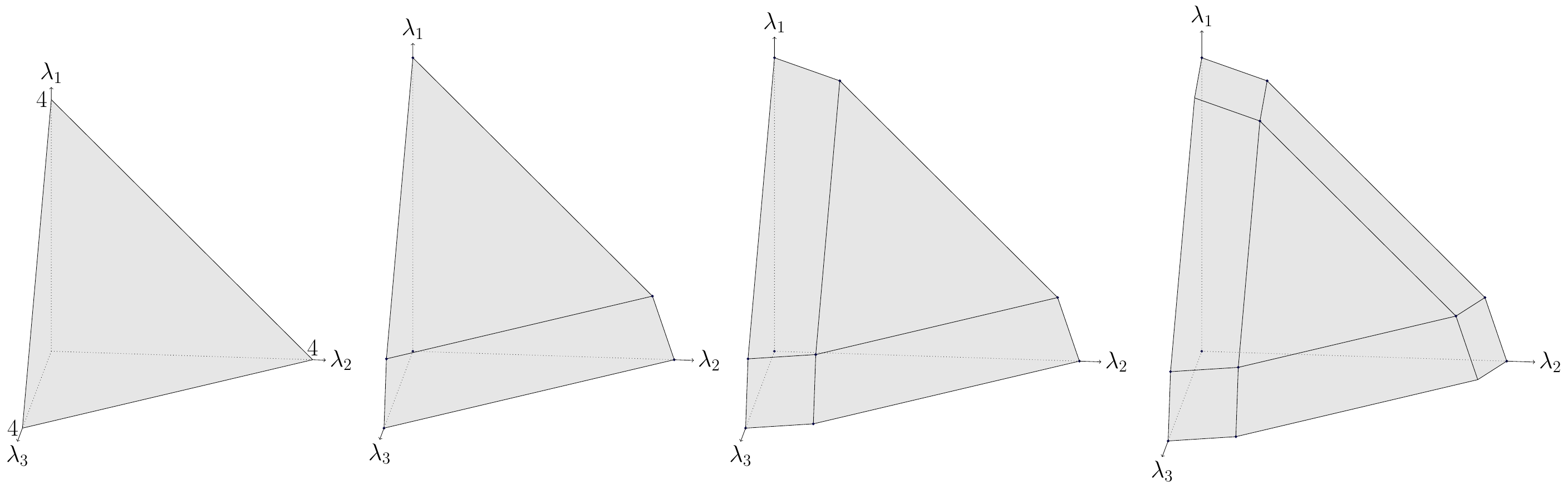}
    \caption{Rate polytopes \(\mathcal{S}_i(12,3)\) for \(i=0,1,2,3\). As \(i\) increases, the SRR region strictly expands.}
    \label{fig:G123_RateRegion}
\end{figure}

Figure~\ref{fig:G123_RateRegion} demonstrates this increasing inclusion for cases \(n=12\) and \(k=3\). We observe that increasing the number of systematic columns expands the service region in storage systems employing MDS codes with the same number of servers. The list of linear constraints that characterize these regions are derived in Section~\ref{Sec:SRR_polytopes}.


\section{Greedy Matching Allocation}\label{Sec:Greedy_Matching}
This section introduces an efficient allocation scheme called \emph{Greedy matching}, inspired by the classical Greedy algorithm. The key idea is to serve each data request using its associated systematic server as much as possible, since systematic servers store uncoded copies and can independently serve their corresponding data objects without collaboration. We prove that for any system employing an MDS code with generator matrix \( \boldsymbol{G} \), and any service rate \( \boldsymbol{\lambda} \in \mathcal{S}(\boldsymbol{G}) \), there exists a matching \( \boldsymbol{w} \) such that \( \boldsymbol{\lambda} = \boldsymbol{w}\boldsymbol{S} \), and \( \boldsymbol{w} \) can be realized through Greedy matching, as formally defined below. This allocation strategy plays a central role in characterizing the service regions of MDS codes and offers practical insights for efficient request allocation and load balancing in distributed systems.
\begin{definition}[Greedy Matching for MDS-coded Systems]
\label{def:GreedyMatching}
Consider any vector \(\boldsymbol{\lambda} = (\lambda_1, \ldots, \lambda_k) \, \in \, \mathcal{S}_i\). The Greedy matching scheme allocates as much of each request \(\lambda_j\) as possible to its \emph{systematic} recovery set (i.e., systematic server), while any remaining portion is assigned to \emph{non-systematic} recovery sets. 
Concretely, for each request \(\lambda_j\), \(j \, \in \, [i]\), define
as \(\lambda_j^{\textsf{s}}\) and \(\lambda_j^{\textsf{c}}\) the fraction of \(\lambda_j\) served by the systematic (``\(\textsf{s}\)'') and non-systematic (``\(\textsf{c}\)'') servers, respectively. Then under Greedy matching,
\[
\lambda_j^{\textsf{s}} 
\;=\;
\min\{\lambda_j,\,1\}
\quad\text{and}\quad
\lambda_j^{\textsf{c}} 
\;=\;
\lambda_j - \lambda_j^{\textsf{s}}.
\]

Consequently, all remaining requests 
\[
\sum_{j=1}^{i} \lambda_j^{\textsf{c}} 
\;+\;
\sum_{j=i+1}^k \lambda_j
\;=\;
\sum_{j=1}^{i} \bigl(\lambda_j - 1\bigr)^+
\;+\;
\sum_{j=i+1}^k \lambda_j
\]
will be served by non-systematic nodes and any unused portion of the systematic nodes.
\end{definition}

\begin{theorem}
\label{thm:Greedy}
Any \(\boldsymbol{\lambda} \, \in \, \srr \) can be served via Greedy matching.
\end{theorem}

\begin{proof}
Let \(\boldsymbol{G} = \boldsymbol{G}_i(n,k)\) be fixed, and denote \(\Gamma_{G}\) as its associated recovery hypergraph. Because \(\boldsymbol{\lambda}\) is servable by \(\boldsymbol{G}\), by Proposition~\ref{prop:demand_to_matching} there exists a matching \(M_1\) in \(\Gamma_G\) such that the total weight of the hyperedges labeled by \(\boldsymbol{e}_l\) is \(\lambda_l\) for each \(l \, \in \, [k]\), i.e.,\ \(\boldsymbol{\lambda}(M_1) = M_1 \boldsymbol{S}\). Suppose $M_1$ is not greedy, i.e., there is some \(l \, \in \, [i]\) for which the portion of \(\lambda_l\) served by its systematic column is
\[
\lambda_l^{\mathsf{s}}(M_1) 
\;<\; 
\min\{1,\;\lambda_l\}.
\]
We will construct another matching \(M_2\) with
\begin{equation}
\label{eq:Greedycondition}
\begin{cases}
\lambda_l^{\mathsf{s}}(M_2) \;>\; \lambda_l^{\mathsf{s}}(M_1),\\
\lambda_j^{\mathsf{s}}(M_2) \;=\; \lambda_j^{\mathsf{s}}(M_1), & \forall \,\, j \neq l,\ j \, \in \, [i].
\end{cases}
\end{equation}
We consider two cases:

\paragraph{Case 1: The systematic server \(G^l = \boldsymbol{e}_l\) is not saturated.}
Because \(\lambda_l^{\mathsf{s}}(M_1) < \min\{1,\lambda_l\}\), some of \(\lambda_l\) is allocated to non-systematic edges instead. Since \(G^l\) is not fully utilized, we can shift a small portion \(\delta>0\) of \(\lambda_l\) onto \(G^l\). On the other hand, choose any hyperedge serving \(\lambda_l^{\mathsf{c}}\) and reduce its weight by \(\delta\). 
Because \(G^l\) was not saturated, no node in the system becomes overused after the shift, preserving feasibility. Thus, we obtain a valid matching \(M_2\) with
\[
\lambda_l^{\mathsf{s}}(M_2) 
\;=\;
\lambda_l^{\mathsf{s}}(M_1) + \delta
\;>\;
\lambda_l^{\mathsf{s}}(M_1),
\]
and \(\lambda_j^{\mathsf{s}}(M_2) = \lambda_j^{\mathsf{s}}(M_1)\) for all \(j \ne l\).

\paragraph{Case 2: The systematic column \(G^l\) is saturated.}
Although \( G^l \) is fully occupied, the condition \( \lambda_l^{\mathsf{s}}(M_1) < \min\{1, \lambda_l\} \) implies that some portion of \( G^l \) is being used to recover a different vector \( \boldsymbol{e}_m \) with \( m \neq l \). Let \( \epsilon \) be a non-systematic edge containing \( G^l \) and labeled by \( \boldsymbol{e}_m \), and let \( w = \lambda_{m,\epsilon} \) be its weight. 

Similarly, the condition \( \lambda_l^{\mathsf{s}}(M_1) < \min\{1, \lambda_l\} \) also implies the existence of a non-systematic edge \( \epsilon_1 \) (i.e., \( \epsilon_1 \not\ni G^l \)) that contributes to the recovery of \( \lambda_l \). Let its weight be \( \lambda_{l,\epsilon_1} = w_1 \). Since \( \epsilon \) and \( \epsilon_1 \) differ in at least one node—one contains \( G^l \), while the other does not—it follows that \( |\epsilon \,\cup\, \epsilon_1| \geq k+1 \).

\begin{enumerate}
\item If \(w_1 \le w\), then we reduce by \(w_1\) the weight on both \(\epsilon\) and \(\epsilon_1\). We reassign that newly released \(w_1\) portion of \(G^l\) to \(\lambda_l\). Simultaneously, we use a newly freed \(k\)-subset of nodes (excluding \(G^l\), we can do that since $|\epsilon \,\cup\, \epsilon_1| \ge k+1$) to serve \(\boldsymbol{e}_m\). This yields a valid matching \(M_2\) satisfying Eq.~\eqref{eq:Greedycondition}.
\item If \(w_1 > w\), we release \(w\) amount from \(\epsilon\) and \(\epsilon_1\) instead, then apply the same argument.
\end{enumerate}

In either case, we can incrementally increase \(\lambda_l^{\mathsf{s}}(M)\) while maintaining a valid matching. Repeating this argument shows that we can continue shifting capacity until
\[
\lambda_l^{\mathsf{s}}(M_s) 
\;=\;
\min\{1,\;\lambda_l\}
\]
for some matching \(M_s\). Thus, all of \(\lambda_l\) is served systematically (if \(\lambda_l \le 1\)) or the systematic node is fully saturated by \(\lambda_l\) (if \(\lambda_l \ge 1\)).

Hence, any vector \(\boldsymbol{\lambda}\) achievable under any allocation can be served by a Greedy matching that prioritizes each systematic column up to \(\min\{\lambda_l,1\}\). This proves the claimed result: sending each request to its systematic node until it is served entirely or that node is saturated does not compromise feasibility. 
\end{proof}
Mathematically, Greedy matching allows us to pre-assign values to certain edge weights in the vector $\boldsymbol{w}$ before checking the feasibility of the linear program in Remark~\ref{rm:feasibility_check}. This pre-assignment \emph{reduces the degrees of freedom}, and therefore substantially reduces the feasibility check's complexity. Concretely, for a recovery hypergraph $\Gamma_i$ of $\boldsymbol{G}_i$ and a demand $\lambda_j$ such that $j \le i$ (i.e., $\boldsymbol{G}_i$ contains a systematic node for object $j$):
\begin{itemize}
    \item If $\lambda_j \ge 1$, then a weight $w = 1$ is assigned to the systematic edge labeled $\boldsymbol{e}_j$. Since the corresponding server is fully utilized, all other recovery sets that contain this node $\boldsymbol{e}_j$ are assigned weight 0 and thus can be removed.
    \item If $\lambda_j < 1$, then a weight $w = \lambda_j$ is assigned to the systematic edge labeled $\boldsymbol{e}_j$. Since the demand for object $j$ is fully served, all recovery sets labeled by $\boldsymbol{e}_j$ are assigned weight 0 and can be removed.
\end{itemize}
This process is demonstrated in Figure~\ref{fig:RG-42-SYS}. Furthermore, in Theorem~\ref{thm:greedy_optimal} in the final section, we show cases where Greedy matching is \emph{the only possible} rate-splitting scheme to achieve some particular request vectors in the SRR.

\begin{remark}
A special case of this theorem where $\boldsymbol{G}$ is systematic (i.e., $i = k$), was proven in~\cite{Service:journals/tit/AktasJKKS21}, Lemmas 1 and 2. The authors established that sending requests first to their corresponding systematic node is optimal. Based on this result, they devised a Water-filling algorithm for request rate splitting. Once the systematic servers are saturated, the requests are sent to the $k$ currently least-loaded servers, which collaboratively form a non-systematic recovery set.
\end{remark}

\begin{figure}[hbt]
    \centering
     \includegraphics[scale=0.85]{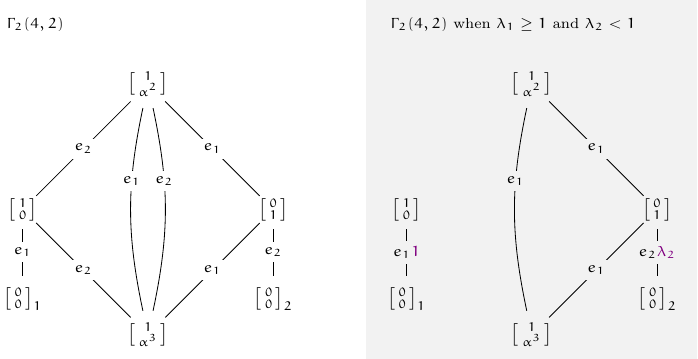}
    \caption{Recovery graph $\Gamma_2(4,2)$ (left) and Greedy matching on $\Gamma_2(4,2)$ when $\lambda_1 \ge 1$ and $\lambda_2 < 1$ (right). After the Greedy matching, the systematic edge labeled $\boldsymbol{e}_2$ is assigned a weight of $\lambda_2 < 1$, while all other edges labeled $\boldsymbol{e}_2$ are assigned a weight of 0 and thus removed. The systematic edge labeled $\boldsymbol{e}_1$ is assigned a weight of 1. The degrees of freedom, which is the number of edges without a weight, are reduced from 8 to 3.}
    \label{fig:RG-42-SYS}
\end{figure}
\section{Characterizing the SRR Polytopes}\label{Sec:SRR_polytopes}

We now develop a linear characterization of the SRRs for MDS-coded systems, providing exact conditions under which data requests can be served. This characterization builds on the optimality of Greedy Matching allocation presented earlier, and is the first explicit formulation of the SRR for MDS codes. We begin by recalling that \(\mathcal{S}_0\), the SRR polytope of \(\boldsymbol{G}_0(n,k)\), is described by the following \(k + 1\) linear constraints:
\[
\begin{cases}
\sum_{j=1}^k \lambda_j \;\le\; {n}/{k}, & \quad (\text{1 constraint}), \\[4pt]
\lambda_j \;\ge\; 0, \quad \forall \,\, j \, \in \, [k], & \quad (k \text{ constraints}).
\end{cases}
\]
The service polytope of $\boldsymbol{G}_0(6,3)$, which is a simplex in the positive orthant of $\mathbb{R}^3$ given by $\sum_{i=1}^3\lambda_i \le 2$, is shown in Fig.~\ref{fig:SRR_simplex_63}.
\begin{figure}
    \centering
    \includegraphics[width=0.31\linewidth]{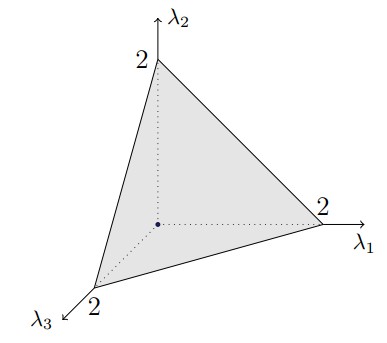}
    \caption{The service polytope of $\boldsymbol{G}_0(6,3)$.}
    \label{fig:SRR_simplex_63}
\end{figure}
\subsection{Bounding the SRR via Fractional Matchings}

We now generalize this description to the SRRs corresponding to other generator matrices \(\boldsymbol{G}_i(n,k)\). We begin by recalling a useful lemma from~\cite{DualCode_bound:conf/isit/AlfaranoRS22}:


\begin{lemma}[\hspace{-0.3mm}\cite{DualCode_bound:conf/isit/AlfaranoRS22}, Lemma IV.1]\label{lemma:MatchingBound}
For any fractional matching \(\boldsymbol{w} \, \in \, \mathbb{R}^{|E|}\) on a hypergraph \((V, E)\), the following inequality holds:
\[
|V| \;\ge\; \sum_{\epsilon \, \in \, E} w_{\epsilon}\,|\epsilon|,
\]
where $|\epsilon|$ denotes the size (or cardinality) of the hyperedge $\epsilon$. This inequality is referred to as the \emph{capacity bound}, where the left-hand side represents the total capacity of all nodes. A matching that satisfies this bound with equality is called \emph{perfect}.
\end{lemma}

This result helps characterize the SRR polytopes of MDS codes described in Section~\ref{Sec:Inclusion}. Specifically, by leveraging the Greedy matching approach (Section~\ref{Sec:Greedy_Matching}), we show that the set of achievable demand vectors \(\boldsymbol{\lambda} \, \in \, \mathcal{S}_i\) must lie within certain linear bounds derived from Lemma~\ref{lemma:MatchingBound}. We then prove that these bounds can be attained in many cases, thereby fully characterizing the corresponding SRR polytopes.

\medskip
\noindent
\textbf{Greedy Allocation and Fractional-Matching Bound.}
Consider a coding scheme defined by \(\boldsymbol{G}_i(n,k)\), and let \(\boldsymbol{\lambda} = (\lambda_1, \dots, \lambda_k) \, \in \, \mathcal{S}_i\) be any achievable demand vector. Without loss of generality, assume \(\lambda_1 \ge \lambda_2 \ge \dots \ge \lambda_k\). Define $i_A$ as the unique number in $[i]$ so that
\[
\lambda_j \;\ge\; 1 \quad \text{for all } j \le i_A,
\quad\text{and}\quad 
\lambda_j \;<\; 1 \quad \text{for } i_A < j \le i.
\]
Under the Greedy matching scheme (see Section~\ref{Sec:Greedy_Matching}), we allocate:
\begin{enumerate}
    \item \(\lambda_j^{\mathsf{s}} = 1\) for each \(j \le i_A\),  
    \item \(\lambda_j^{\mathsf{s}} = \lambda_j\) for each \(i_A < j \le i\),  
    \item \(\lambda_j^{\mathsf{s}} = 0\) for each \(j > i\).  
\end{enumerate}
Hence, in the Greedy matching step, each systematic server associated with a basis vector $\boldsymbol{e}_j$, \( j \leq i_A \) is fully saturated, while servers associated with basis vectors $\boldsymbol{e}_j$, \( i_A < j \leq i \) are utilized up to \( \lambda_j \), where \( \lambda_j < 1 \). All remaining demands
\[
\sum_{j=1}^{i_A} (\lambda_j - 1)
\;+\;
\sum_{j=i+1}^k \lambda_j
\]
must be served by non-systematic edges of size \(k\). By Lemma~\ref{lemma:MatchingBound}, these edges impose the following constraint (note that the cardinality of all non-systematic recovery sets are $k$):
\begin{align*}
\left(i_A + \sum_{j=i_A+1}^{i}\lambda_j\right) + k\cdot \left(\sum_{j=1}^{i_A}(\lambda_j - 1)
\;+\;
\sum_{j=i+1}^{k}\lambda_j\right)
\;\;\le\;\;
n ,
\end{align*}
which can be written as:
\begin{align}\label{eq:lambdaRemain:define}
k\sum_{j=1}^{i_A}\lambda_j
\;+\;
\sum_{j=i_A+1}^{i}\lambda_j
\;+\;
k\sum_{j=i+1}^k \lambda_j
\;\le\;
n + k - i_A.
\end{align}
We now see that this bound can be \emph{tight}, as the next result shows.

\begin{theorem}[Subgraph Slicing]\label{theo:FWMatching}
If \(n - i \ge k\) then the bound in \eqref{eq:lambdaRemain:define} is achievable.
\end{theorem}
\begin{proof} We proceed through the following steps.

\textbf{Post-Greedy Set-Up.}
Under Greedy Matching, each systematic server associated with a basis vector $\boldsymbol{e}_j$, \( j \leq i_A \) is fully saturated. Those associated with $\boldsymbol{e}_j$, \(i_A < j \le i\) were partially used, leaving capacity \(1 - \lambda_j\). Finally, any node beyond node \(i\) remains fully available (capacity \(1\)). Formally, the capacity (``Cap") of node \(G^j\) after Greedy is
\[
\text{Cap}(G^j)
\;=\;
\begin{cases}
0, & j \le i_A,\\
1 - \lambda_j, & i_A < j \le i,\\
1, & j > i.
\end{cases}
\]
Because \(n - i \ge k\), at least \(k\) columns (nodes) remain entirely unused, which will be crucial to constructing further \(k\)-sized matchings.

\textbf{Residual Hypergraph $H$.}
Let \(H\) denote the \emph{residual} hypergraph after the Greedy step, where each node has the above capacity. We want to show that the leftover demands
\[
\sum_{j=1}^{i_A}(\lambda_j - 1)
\;+\;
\sum_{j=i+1}^k \lambda_j
\]
are exactly matched by size-\(k\) edges in \(H\). However, node capacities in \(H\) are non-uniform: some have capacity \(1-\lambda_j\), others have capacity \(1\). To handle this, we ``slice'' \(H\) into smaller \(k\)-uniform hypergraphs whose nodes each have a \emph{uniformly assigned} capacity.

\textbf{Slicing Construction.}
We form \((i - i_A + 1)\) sub-hypergraphs \(H_1,\dots,H_{i-i_A+1}\), each resembling a \(k\)-uniform hypergraph on some subset of columns:
\begin{itemize}
    \item \(\displaystyle H_1\) uses \((n - i_A)\) columns each at capacity 
          \(\alpha_1 = 1 - \lambda_{i_A+1}\)
          (or \(\alpha_1 = 1\) if \(i_A = i\)).
    \item \(\displaystyle H_2\) has \((n - i_A - 1)\) columns each at capacity 
          \(\alpha_2 = \lambda_{i_A+1} - \lambda_{i_A+2}\),
          and so on.
    \item In general, each \(H_\ell\) has uniform capacity 
          \(\alpha_\ell = \lambda_{i_A+\ell-1} - \lambda_{i_A+\ell}\) 
          among a progressively smaller set of columns.
    \item \(\displaystyle H_{\,i-i_A+1}\) (the final slice) has 
          \(\alpha_{\,i-i_A+1} = \lambda_i\)
          among \((n-i)\) columns.
\end{itemize}
(If some \(\lambda_{j}\) are zero for \(j>i\) or \(j \le i_A\), we adjust accordingly, but the concept remains the same.)

\textbf{Capacity Summation.}
Since \( n - i_A \geq n - i \geq k \), each \( H_\ell \) is essentially a \( k \)-uniform hypergraph on \( m_\ell \) vertices (columns), where \( m_\ell \) denotes the number of vertices in \( H_\ell \) and satisfies \( m_\ell \geq k \). By the known \(\mathcal{S}_0\) bound
\(\sum_{j=1}^k \lambda_j \le \tfrac{m_\ell}{k}\),
we can serve exactly 
\(\alpha_\ell\,\tfrac{m_\ell}{k}\)
units of demand from the slice \(H_\ell\). Summing over slices, the total capacity becomes
\begin{align*}
    & (1-\lambda_{i_A+1})\dfrac{n-i_A}{k} + (\lambda_{i_A+1}-\lambda_{i_A+2})\dfrac{n-i_A-1}{k} + \cdots + (\lambda_{i-1}-\lambda_{i})\dfrac{n-i+1}{k} + \lambda_{i}\dfrac{n-i}{k}\\
    & = \dfrac{n-i_A}{k} - \dfrac{1}{k}\sum\limits_{j=i_A+1}^i\lambda_j.    
\end{align*}

\textbf{Matching the Leftover Demands Exactly.}
The leftover demand after Greedy is
\[
\sum_{j=1}^{i_A} (\lambda_j - 1)
\;+\;
\sum_{j=i+1}^k \lambda_j.
\]
Because the slices \(H_\ell\) are \(k\)-uniform sub-hypergraphs with uniform capacity \(\alpha_\ell\), each can form a fractional (or integral) matching that fully utilizes its \(\alpha_\ell\,\tfrac{m_\ell}{k}\) capacity. 
By choosing appropriate edges (of size \(k\)) in each sub-hypergraph, we allocate these capacities to meet the leftover demands perfectly. Hence, equality
\[
\sum_{j=1}^{i_A} (\lambda_j - 1)
\;+\;
\sum_{j=i+1}^k \lambda_j = \dfrac{n-i_A}{k} - \dfrac{1}{k}\sum\limits_{j=i_A+1}^i\lambda_j.
\]
can be \emph{attained} in the residual hypergraph \(H\). In other words, Eq.~\eqref{eq:lambdaRemain:define} can be achieved with equality.

\textbf{Conclusion.}
Thus, by constructing the sub-hypergraphs \(H_1,\dots,H_{\,i-i_A+1}\) and allocating the leftover demands slice by slice, we show that the bound in Eq.~\eqref{eq:lambdaRemain:define} is exactly achievable when \(n - i \ge k\). This completes the proof.
\end{proof}
\begin{figure}[hbt]
\centering
 \includegraphics[scale=1]{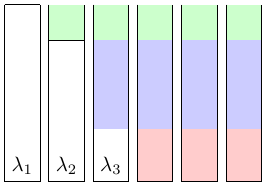}
\caption{Illustration of a storage system \( \boldsymbol{G}_3(6,3) \). In this example, \( \lambda_1 > 1 \), \( \lambda_2 = 0.8 \), and \( \lambda_3 = 0.3 \), so \( i_A = 1 \). After Greedy matching, the hypergraph \( H \) is ``sliced'' into three smaller \( 3 \)-uniform hypergraphs: the first (green) spans the last five vertices with \( \alpha_1 = 1 - \lambda_2 = 0.2 \), the second (purple) spans the last four vertices with \( \alpha_2 = \lambda_2 - \lambda_3 = 0.5 \), and the third (red) spans the last three vertices with \( \alpha_3 = \lambda_3 = 0.3 \). The maximum additional request for $\lambda_1$ that the system can serve after the Greedy matching is therefore constrained by \ensuremath{\lambda_1 - 1 \le \alpha_1\cdot5/3 + \alpha_2\cdot4/3 + \alpha_3\cdot3/3 = 1.3}.}
\label{fig:slicing}
\end{figure}

Figure~\ref{fig:slicing} illustrates this subgraph slicing process based on the values of \( \lambda_1, \lambda_2, \) and \( \lambda_3 \).
\begin{remark}
Eq.~\eqref{eq:lambdaRemain:define} can be equivalently rewritten as
\begin{equation}\label{eq:generalBound}
\sum_{j=1}^{i_A} \lambda_j 
\;+\;
\sum_{j=i+1}^{k} \lambda_j 
\;-\;
i_A
\;\;=\;\;
\frac{n - i_A - \sum_{j=i_A+1}^{i}\lambda_j}{k}.
\end{equation}
When \(i_A = k\) (meaning $\boldsymbol{G} = \boldsymbol{G}_k$ is a systematic matrix), this recovers Theorem~1 of \cite{Service:journals/tit/AktasJKKS21} as a special case. In other words, \eqref{eq:generalBound} generalizes that result to systems where some (but not all) columns are systematic, showing that similarly tight bounds hold for both systematic and partially systematic codes.
\end{remark}
We move on to characterize the SRRs in some specific cases.

\subsection{$\mathcal{S}_i(n, k)$ when \(n \ge k + i\)}
\label{sec:case1-n-gi}

In this scenario, \(n - i \ge k\). By Theorem~\ref{theo:FWMatching}, the bound therein is achieved, implying 
\begin{equation}\label{eq:polytopeDerive}
    k \Bigl(\sum_{j=1}^{i_A}\lambda_j + \sum_{j=i+1}^{k}\lambda_j \Bigr)
    \;+\;
    \sum_{j=i_A+1}^{i}\lambda_j
    \;=\;
    n + i_A(k - 1).
\end{equation}

Let us define a partition of the index set \(\{1, 2, \ldots, i\}\) into three subsets \(\mathcal\mathcal{A}, \mathcal\mathcal{B}, \mathcal\mathcal{C}\) as follows:
\[
\mathcal\mathcal{A} \;=\;\{\,j \, \in \, [i] : \lambda_j \ge 1\},
\qquad
\mathcal\mathcal{B} \;=\; [i]\setminus\mathcal\mathcal{A},
\qquad
\mathcal\mathcal{C} \;=\;\{\,i+1, i+2, \dots, k\}.
\]

Then \eqref{eq:polytopeDerive} can be rewritten as
\[
k\Bigl(\sum_{j\, \in \, \mathcal\mathcal{A}} \lambda_j + \sum_{j\, \in \, \mathcal\mathcal{C}} \lambda_j\Bigr) 
\;+\;
\sum_{j \, \in \, \mathcal\mathcal{B}} \lambda_j
\;=\;
n \;+\; |\mathcal\mathcal{A}|\,(k - 1) = n+a(k-1),
\]

where we denote $a = |\mathcal{A}|$. Because any choice of \(\mathcal\mathcal{A}\subseteq [i]\) uniquely determines \(\mathcal\mathcal{B} = [i]\setminus \mathcal\mathcal{A}\), and there are \(2^i\) ways to pick \(\mathcal\mathcal{A}\subseteq [i]\), we obtain a family of linear constraints describing \(\mathcal{S}_i\). Concretely, for \(\boldsymbol{G}_i(n,k)\) with \(n \ge k + i\), the SRR \(\mathcal{S}_i\) satisfies:

\begin{theorem}\label{theo:polytope}
If \(n \ge k + i\), the service region \(\mathcal{S}_i\) is given by the set of all nonnegative \(\boldsymbol{\lambda}\) satisfying:
\begin{equation}\label{equation1}
\begin{cases}
\lambda_l \ge 0, & \forall \,\, l \, \in \, [k], \\[4pt]
k\,\Bigl(\sum_{j=1}^k \lambda_j\Bigr) \;\le\; n + i\,(k-1), & \\[4pt]
\lambda_l \;+\; k\Bigl(\sum_{j \, \in \, [k]\setminus \{l\}} \lambda_j\Bigr) 
\;\le\; n \;+\; (i-1)\,(k-1), 
& \forall \,\,l \, \in \, [i],\\[4pt]
\lambda_l \;+\; \lambda_h \;+\; 
k\Bigl(\sum_{j \, \in \, [k]\setminus \{l,h\}} \lambda_j\Bigr)
\;\le\; n \;+\; (i-2)\,(k-1), 
& \forall \,\, l,h \, \in \, [i], \; l \ne h, \\[2pt]
\quad \vdots & \\[2pt]
\sum_{l \, \in \, [i]} \lambda_l \;+\; 
k \Bigl(\sum_{j \, \in \, [k]\setminus [i]} \lambda_j\Bigr)
\;\le\; n. &
\end{cases}
\end{equation}
\end{theorem}

\noindent
There are \( k + 2^i \) constraints in total: \( k \) nonnegativity constraints and \( 2^i \) constraints arising from all subsets \( \mathcal\mathcal{A} \subseteq [i] \). In the special case \( i = k \), the second and last constraints in \eqref{equation1} become linearly dependent. Thus, only the tighter one—the last constraint—remains effective, resulting in \( 2^k + k - 1 \) distinct constraints.
\begin{remark}
If \( n \geq 2k \), the condition \( n \geq k + i \) is always satisfied since \( i \leq k \). Therefore, for any MDS-coded system with at least twice as many servers as data objects, the SRR is fully determined as in Theorem~\ref{theo:polytope} regardless of the number of systematic servers. 
\end{remark}
\subsection{$\mathcal{S}_k(k+1, k)$ of systematic coding}
Now consider the matrix \(\boldsymbol{G} = \boldsymbol{G}_k(k+1,k)\), which corresponds to a systematic code with \(k\) systematic columns and 1 parity column $p$. The code generated by this generator matrix is called Single-parity check (SPC) code. Hence, \(\mathcal\mathcal{C} = \emptyset\) and the partition reduces to \(\mathcal\mathcal{A}, \mathcal\mathcal{B}\subseteq [k]\) with \(\mathcal\mathcal{A}\,\cup\,\mathcal\mathcal{B} = [k]\). Then \eqref{eq:lambdaRemain:define} simplifies to
\begin{equation}\label{eq:k+1}
    k\,\Bigl(\sum_{j \, \in \, \mathcal\mathcal{A}} \lambda_j \;-\; |\mathcal\mathcal{A}|\Bigr) 
    \;+\;
    \sum_{j \, \in \, \mathcal\mathcal{B}} \lambda_j
    \;\;\le\;\; (k+1) \;-\; |\mathcal\mathcal{A}|.
\end{equation}

We show that in this case, each pair of distinct requests \(\lambda_i,\lambda_j\) must satisfy \(\lambda_i + \lambda_j \le 2\). Formally:

\begin{lemma}
For any data request \(\boldsymbol{\lambda}\, \in \, \mathcal{S}_k(k+1,k)\), one has
\[
\lambda_i \;+\; \lambda_j \;\le\; 2 
\quad
\forall \,\, i\neq j,\; i,j\, \in \, [k].
\]
\end{lemma}

\begin{proof}
In the recovery graph \( \Gamma_G \), consider the subgraph \( \Gamma' \) obtained by removing all edges except those labeled by \( \boldsymbol{e}_i \) and \( \boldsymbol{e}_j \). Assign weight 1 to the two nodes corresponding to the systematic columns for \( i \) and \( j \), and weight 0 to all other nodes. Since the remaining \( (k-1) \) nodes cannot independently form a recovery set for either \( \boldsymbol{e}_i \) or \( \boldsymbol{e}_j \) without involving their respective systematic columns, this weight assignment forms a valid vertex cover of $\Gamma'$ of size 2. Consequently, Proposition~\ref{prop:sum_bound} implies \( \lambda_i + \lambda_j \leq 2 \).
\end{proof}

We argue that this pairwise constraint, along with \eqref{eq:k+1} and nonnegativity, fully characterizes the SRR polytope. Concretely, we have the following result.

\begin{theorem}
The service region $\mathcal{S}_k(k+1,k)$ is given by
\begin{align}
    \lambda_j &\;\ge\; 0, 
    \quad \forall \,\, j \, \in \, [k], 
    && \text{(nonnegativity)} \label{eq:non-negative-k+1}\\[4pt]
    \lambda_i + \lambda_j &\;\le\; 2, 
    \quad \forall \,\, i \neq j,\; i,j \, \in \, [k],
    && \text{(vertex-cover constraints)} \label{eq:vertex-cover-k+1}\\[4pt]
    k\,\Bigl(\sum_{j \, \in \, \mathcal\mathcal{A}}\lambda_j \;-\; |\mathcal\mathcal{A}|\Bigr)
    + \sum_{j \, \in \, \mathcal\mathcal{B}} \lambda_j 
    &\;\le\; k+1 - |\mathcal\mathcal{A}|,
    && \text{(node-capacity constraint)} \label{eq:node-capacity-k+1}
\end{align}
where \(\mathcal\mathcal{A}\) is the set of indices \(j\) with \(\lambda_j > 1\) and \(\{\mathcal\mathcal{A},\mathcal\mathcal{B}\}\) forms a partition of \([k]\).
\end{theorem}

\begin{proof}
Because \( \boldsymbol{G} \) includes systematic columns for all \( k \) objects, the roles of \( \lambda_1, \dots, \lambda_k \) are symmetric. Without loss of generality, let \( \lambda_1 \geq \lambda_2 \geq \dots \geq \lambda_k \). Consider the constraint 
\(\lambda_1 + \lambda_2 \leq 2\), which implies that only \( \lambda_1 \) can exceed 1.  

\paragraph{Case 1: \( \lambda_1 \leq 1 \)}
In this case, \( \lambda_j \leq 1 \) for all \( j \). This request is clearly servable via Greedy matching, i.e., each request is served by its corresponding systematic server.

\paragraph{Case 2: \( \lambda_1 > 1 \).}
Let \( \lambda_1 = 1 + \delta \) with \( \delta > 0 \). From \( \lambda_1 + \lambda_2 \leq 2 \), it follows that \( \lambda_2 \leq 1 - \delta \). Thus, we have \( \lambda_1 > 1 > \lambda_2 \geq \dots \geq \lambda_k \). 

After Greedy matching, requests \( \lambda_2, \dots, \lambda_k \) are fully served by their corresponding systematic columns. The remaining portion, \( \lambda_1 - 1 = \delta \), can be served using a non-systematic recovery set consisting of:
\begin{itemize}
\item the free capacity in \( (k-1) \) systematic nodes \( \boldsymbol{e}_2, \dots, \boldsymbol{e}_k \) (each with at least \( \delta \) available), and
\item the single unused parity column \( p \) in \( \boldsymbol{G}_k \).
\end{itemize}
Thus, \( \lambda_1 \) is fully servable. These arguments show that the conditions \eqref{eq:non-negative-k+1}, \eqref{eq:vertex-cover-k+1}, and \eqref{eq:node-capacity-k+1} indeed describe \(\mathcal{S}_k(k+1,k)\) precisely.
\end{proof}
\subsection{$\mathcal{S}_i(n, k)$ when $n = k+i-1$}
In this case $\boldsymbol{G} = \boldsymbol{G}_i(k+i-1, k), \mathcal{C} = [k]\setminus[i]$ and $\mathcal{A}\,\cup\,\mathcal{B} = [i]$, therefore Eq.~\eqref{eq:generalBound} becomes:
\begin{align}
    k\Bigl(\sum\limits_{j \, \in \, \mathcal{A} \,\cup\, \mathcal{C}} \lambda_j - |\mathcal{A}|\Bigr) + \sum\limits_{j \, \in \, \mathcal{B}} \lambda_j & \le k+i-1 - |\mathcal{A}|
    \label{eq:k+i-1}
\end{align}

This corner case is substantially more complicated than the previous cases when $n \ge k+i$. 
We present and prove the following constraint:
\begin{lemma} For any service vector $\boldsymbol{\lambda} \, \in \, \mathcal{S}_i(k+i-1, k)$: 
\begin{equation}
    \sum\limits_{j \, \in \, \mathcal{A} \,\cup\, \mathcal{C}} \lambda_j + \sum\limits_{j \, \in \, \mathcal{B}} \lambda_j = \sum\limits_{j=1}^k \lambda_j \le i
\end{equation}
\label{lem:sum_rate_bound}
\end{lemma}
\begin{proof}
    Put a weight one on $i$ systematic vertices of the recovery graph $\Gamma_G$;  the remaining \( (k-1) \) nodes cannot independently form a non-systematic recovery set for any object without involving their respective systematic columns. Thus, we have a valid vertex cover of size $i$, by Proposition~\ref{prop:sum_bound}, $i \ge \sum\limits_{j=1}^k \lambda_j$.
\end{proof}
We will prove that this constraint and constraint Eq.~\eqref{eq:k+i-1}, along with non-negativity constraints characterize the SRR polytope.
\begin{theorem}
$\mathcal{S}_i(k+i-1,k)$ is given by:
\begin{align}
    \lambda_j &\ge 0, \quad \forall \, j \, \in \, [k] 
    && \text{(Non-negativity constraints)} \label{eq:non-negative} \\
    \sum\limits_{j=1}^{k} \lambda_j &\le i 
    && \text{(Vertex cover constraint)} \label{eq:vertex:cover} \\  
    k\sum\limits_{j \, \in \, \mathcal{A}} (\lambda_j - 1) + k\sum\limits_{j \, \in \, \mathcal{C}} \lambda_j + \sum\limits_{j \, \in \, \mathcal{B}} \lambda_j &\le k + i - 1 - |\mathcal{A}|
    && \text{(Node capacity)} \label{eq:node:capacity}
\end{align}
    
\end{theorem}

\begin{proof}
    Let $a = |\mathcal{A}|$. For any non-negative request vector $\boldsymbol{\lambda} = (\lambda_1, \lambda_2, \hdots, \lambda_k)$ that satisfies \eqref{eq:vertex:cover} and \eqref{eq:node:capacity}, we prove that it is servable by the system by considering one of the following 3 cases:
\begin{enumerate}
    \item The Vertex cover constraint satisfied with equality:
\begin{equation}
\begin{cases}
    & k\sum\limits_{j \, \in \, \mathcal{A}}(\lambda_j - 1) + k\sum\limits_{j \, \in \, \mathcal{C}}\lambda_j +
    \sum\limits_{j \, \in \, \mathcal{B}}\lambda_j \le k+i-1-a\\
    & \sum\limits_{j=1}^{k}\lambda_j = i \label{eq:n=k+i-1:firstcase}
\end{cases}
\end{equation}
We first serve this request $\boldsymbol{\lambda}$ using Greedy matching. After Greedy matching allocation, the remaining request to be served is $\sum\limits_{i \, \in \, \mathcal{A}}(\lambda_i-1) + \sum\limits_{j \, \in \, \mathcal{C}}\lambda_j$. Note that because $n = k+i-1$, $k-1$ columns (nodes) remain entirely unused.

We now prove that $\sum\limits_{j \, \in \, \mathcal{A}}(\lambda_j-1) + \sum\limits_{j \, \in \, \mathcal{C}}\lambda_j \le 1$, i.e., $\sum\limits_{j \, \in \, \mathcal{A} \,\cup\, \mathcal{C}}\lambda_j \le a+1$, where $a = |\mathcal{A}|$. Indeed, assume otherwise that $\sum\limits_{j \, \in \, \mathcal{A} \,\cup\, \mathcal{C}}\lambda_j = a + 1 + \delta$ for some $\delta > 0$. From \eqref{eq:n=k+i-1:firstcase}, we have
\begin{align}
k+i-1-a & \ge k\sum\limits_{j \, \in \, \mathcal{A}}(\lambda_j - 1) + k\sum\limits_{j \, \in \, \mathcal{C}}\lambda_j + \sum\limits_{j \, \in \, \mathcal{B}}\lambda_j \notag \\
& = k\Bigl(\sum\limits_{j \, \in \, \mathcal{A}\,\cup\, \mathcal{C}}\lambda_j - a \Bigr) + \sum\limits_{j \, \in \, \mathcal{B}}\lambda_j \notag \\
& = k\delta + k + \Bigl(i - \sum\limits_{j \, \in \, \mathcal{A} \,\cup\, \mathcal{C}}\lambda_j\Bigr) \notag \\
& = k\delta + k+i - (a+1+\delta)
\end{align}
where the second equality comes from the original assumption that $i = \sum_{j=1}^k \lambda_j= \sum\limits_{j \, \in \, \mathcal{A}\,\cup\, \mathcal{C}}\lambda_j + \sum\limits_{j \, \in \, B}\lambda_j$. The last equation leads us to $0 \ge k\delta - \delta$ or $1 \ge k$, which could not happen when $k \ge 2$. Thus by contradiction, we have: $\sum\limits_{j \, \in \, \mathcal{A} \,\cup\, \mathcal{C}}\lambda_j \le a+1$.

Therefore, note that $\mathcal{A} \,\cup\, \mathcal{B} = [i]$, we have
\begin{align*}
1 \ge \sum\limits_{j \, \in \, {\mathcal{A}\,\cup\, \mathcal{C}}}\lambda_j - a = i - \sum\limits_{j \, \in \, \mathcal{B}}\lambda_j - a = |\mathcal{B}| - \sum\limits_{j \, \in \, \mathcal{B}}\lambda_j = \sum\limits_{j \, \in \, \mathcal{B}}(1-\lambda_j) = \sum\limits_{j \, \in \, \mathcal{A}}(\lambda_j - 1) + \sum\limits_{j \, \in \, \mathcal{C}}\lambda_j.
\end{align*}
which means that the sum of the free capacity of nodes in $\mathcal{B}$ after the greedy matching step is at most 1. Under this condition, we prove that the request rate can be served using the ensuing lemma.
\begin{lemma}{(Successive Tiling)}
Under the coding scheme \( \boldsymbol{G}_i(k+i-1, k) \), any request vector \( \boldsymbol{\lambda} = (\lambda_1, \lambda_2, \dots, \lambda_k) \) that satisfies the following constraints is servable:
\begin{equation}
\begin{cases}
    k\sum\limits_{j \, \in \, \mathcal{A}}(\lambda_j - 1) + k\sum\limits_{j \, \in \, \mathcal{C}}\lambda_j + \sum\limits_{j \, \in \, \mathcal{B}}\lambda_j \le k+i-1-a,\\
    \sum\limits_{j \, \in \, \mathcal{A} \,\cup\, \mathcal{C}}\lambda_j + \sum\limits_{j \, \in \, \mathcal{B}}\lambda_j = \sum\limits_{j=1}^{k}\lambda_j = i, \label{eq:n=k+i-1:successive}\\
    \sum\limits_{j \, \in \, \mathcal{A}}(\lambda_j - 1) + \sum\limits_{j \, \in \, \mathcal{C}}\lambda_j \le 1.
\end{cases}
\end{equation}
\end{lemma}

\begin{proof}
We again serve this request using Greedy matching:
\begin{equation}
    \begin{cases}
      \lambda_j^{\mathsf{s}} = 1, \quad \forall \, j \, \in \, \mathcal{A}, \\
      \lambda_j^{\mathsf{s}} = \lambda_j, \quad \forall \, j \, \in \, \mathcal{B}, \label{eq:greedy:successive}\\
      \lambda_j^{\mathsf{s}} = 0, \quad \forall \, j \, \in \, \mathcal{C}.
    \end{cases}
\end{equation}
The remaining request to be served is:
\[
\sum\limits_{j \, \in \, \mathcal{A}}(\lambda_j - 1) + \sum\limits_{j \, \in \, \mathcal{C}}\lambda_j.
\]

After the Greedy matching step, \( k-1 \) non-systematic nodes remain completely unused. For each systematic node \( j \, \in \, \mathcal{B} \), we allocate a \( (1-\lambda_j) \) portion of all the \( k-1 \) non-systematic nodes. Together with the remaining \( (1-\lambda_j) \) portion of systematic node \( j \), this forms a recovery set of size \( k \) (a non-systematic recovery set), which can be used to serve any object with demand \( 1-\lambda_j \).  

From our assumption in \eqref{eq:n=k+i-1:successive}, we have
\[
\sum\limits_{j \, \in \, \mathcal{B}}(1-\lambda_j) = \sum\limits_{j \, \in \, \mathcal{A}}(\lambda_j - 1) + \sum\limits_{j \, \in \, \mathcal{C}}\lambda_j \le 1.
\]
Thus, we can perform this allocation for all systematic nodes \( j \, \in \, \mathcal{B} \) without exhausting the \( k-1 \) non-systematic nodes. In the end, we have \( |\mathcal{B}| \) non-systematic recovery sets of size \( k \) with capacity \( 1-\lambda_j \) for \( j \, \in \, \mathcal{B} \), which can be used to recover any data object. 

Since
\[
\sum\limits_{j \, \in \, \mathcal{B}}(1-\lambda_j) = \sum\limits_{j \, \in \, \mathcal{A}}(\lambda_j-1) + \sum\limits_{j \, \in \, \mathcal{C}}\lambda_j,
\]
all the remaining demand for objects in \( \mathcal{A} \) and \( \mathcal{C} \) can be served. This proves that \( \boldsymbol{\lambda} \) is servable.
\end{proof}
Fig.~\ref{fig:successive_tiling} exemplifies this scenario in $\boldsymbol{G}_3(5, 3)$ where $\boldsymbol{\lambda} = (1.9,\, 0.6,\, 0.5)$. In this example, 
\begin{align*}
\begin{cases}
\mathcal{A} = \{1\},\, \mathcal{B} = \{2, 3\}, \, \mathcal{C} = [3]\setminus\{\mathcal{A} \,\cup\, \mathcal{B}\} = \emptyset\\
\sum\limits_{j \, \in \, \mathcal{A}\,\cup\, \mathcal{C}}\lambda_j+\sum\limits_{j \, \in \, \mathcal{B}}\lambda_j = 3 = i
\end{cases}
\end{align*}
Moreover, $\sum\limits_{j \, \in \, \mathcal{B}} (1 - \lambda_j) = 1 - 0.6 + 1 - 0.5 = 0.9 < 1$.

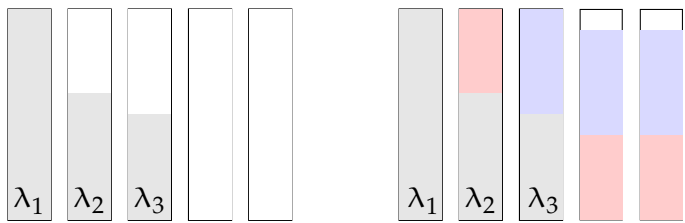
\begin{figure}[hbt]
    \centering
\begin{tikzpicture}[scale=0.4, every node/.style={scale=3, font=\large}]
    \begin{scope}[shift={(0,0)}] 
    \draw[thick] (0,0) rectangle (1.4,7);
    \fill[gray!20] (0,0) rectangle (1.4,7); 
    \node[font=\tiny, anchor=north] at (0.7,2.2) {\scalebox{0.6}{$\lambda_1$}};
    
    \draw[thick] (2,0) rectangle (3.4,7);
    \fill[white!20] (2,4.2) rectangle (3.4,7); 
    \fill[gray!20] (2,0) rectangle (3.4,4.2); 
    \node[font=\tiny\tiny, anchor=north] at (2.7,2.2) {\scalebox{0.6}{$\lambda_2$}};
    
    \draw[thick] (4,0) rectangle (5.4,7);
    \fill[white!15] (4,3.5) rectangle (5.4,7); 
    \fill[gray!20] (4,0) rectangle (5.4,3.5); 
    \node[font=\tiny, anchor=north] at (4.7,2.2) {\scalebox{0.6}{$\lambda_3$}};
    
    \draw[thick] (6,0) rectangle (7.4,7);
    \fill[white!15] (6,0) rectangle (7.4,7); 
    
    \draw[thick] (8,0) rectangle (9.4,7);
    \fill[white!15] (8,0) rectangle (9.4,7); 
    
    \end{scope}

    \begin{scope}[shift={(13,0)}] 
    \draw[thick] (0,0) rectangle (1.4,7);
    \fill[gray!20] (0,0) rectangle (1.4,7); 
    \node[font=\tiny, anchor=north] at (0.8,2.2) {\scalebox{0.6}{$\lambda_1$}};
    
    \draw[thick] (2,0) rectangle (3.4,7);
    \fill[red!20] (2,4.2) rectangle (3.4,7); 
    \fill[gray!20] (2,0) rectangle (3.4,4.2); 
    \node[font=\tiny, anchor=north] at (2.7,2.2) {\scalebox{0.6}{$\lambda_2$}};
    
    \draw[thick] (4,0) rectangle (5.4,7);
    \fill[blue!15] (4,3.5) rectangle (5.4,7); 
    \fill[gray!20] (4,0) rectangle (5.4,3.5); 
    \node[font=\tiny, anchor=north] at (4.8,2.2) {\scalebox{0.6}{$\lambda_3$}};
    
    \draw[] (6,0) rectangle (7.4,7);
    \fill[blue!15] (6,2.8) rectangle (7.4, 6.3); 
    \fill[red!20] (6,0) rectangle (7.4, 2.8); 
    
    \draw[] (8,0) rectangle (9.4,7);
    \fill[blue!15] (8,2.8) rectangle (9.4,6.3); 
    \fill[red!20] (8,0) rectangle (9.4,2.8); 
    \end{scope}

\end{tikzpicture}
\caption{System using $\boldsymbol{G}_3(5, 3)$ with request vector $\boldsymbol{\lambda} = (1.9,\, 0.6,\, 0.5)$. In this example, $\mathcal{A} = \{1\}, \mathcal{B} = \{2, 3\}, \mathcal{C} = \emptyset$, and $\sum_{j=1}^k \lambda_j = 3$. To the left, three systematic nodes were utilized according to Greedy matching. The white portions represent the free capacities remaining after the greedy allocation step. To the right, the remaining capacities of the systematic nodes in $\mathcal{B}$ were used together with the last two non-systematic nodes (associated with two parity columns) to serve the remaining request for $\lambda_1$. Concretely, three red portions, each of capacity 0.4, form a non-systematic recovery set of capacity 0.4, which serves the remaining demand for $\lambda_1$. The same applies to the three blue portions, each of capacity 0.5. Because $(1-\lambda_2) + (1-\lambda_3) < 1$, the non-systematic nodes are not fully utilized (i.e., the associated matching is not perfect), even though this vector satisfies $\sum_{j=1}^k \lambda_j = 3$ and thus lies on the boundary of the SRR.}
\label{fig:successive_tiling}
\end{figure}

\item The Node capacity constraint is tightly satisfied:
\begin{equation}
\begin{cases}
    & k\sum\limits_{j \, \in \, \mathcal{A}}(\lambda_j - 1) + k\sum\limits_{j \, \in \, \mathcal{C}}\lambda_j + \sum\limits_{j \, \in \, \mathcal{B}}\lambda_j = k+i-1-a \label{eq:n=k+i-1:secondcase}\\
    & \sum\limits_{j=1}^{k}\lambda_j < i
\end{cases}
\end{equation}
Observe first that because $\sum\limits_{j=1}^{k}\lambda_j < i$, then there must exists $l \, \in \, [i]$ such that $\lambda_l < 1$ (for otherwise $\sum\limits_{j=1}^{k}\lambda_j \ge \sum\limits_{j=1}^{i}\lambda_j \ge i$, contradiction!). Therefore, $a = |\mathcal{A}| \le i-1$.

We now find a lower bound for $\sum\limits_{j=1}^{k}\lambda_j$. Let $T = \sum\limits_{j \, \in \, \mathcal{B}}\lambda_j$, from \eqref{eq:n=k+i-1:secondcase} we have:
\begin{align}
k\sum\limits_{j \, \in \, \mathcal{A}}(\lambda_j-1) + k\sum\limits_{j \, \in \, \mathcal{C}}\lambda_j & = k+i-1-a-T \\
\Leftrightarrow \quad \sum\limits_{j=1}^{k}\lambda_j = \sum\limits_{j \, \in \, \mathcal{A} \,\cup\, \mathcal{C}}\lambda_j +T & = \dfrac{k+i-1-a-T}{k}+a+T  \\
& = \dfrac{k+i-1-a}{k}+a+\Bigl(T - \dfrac{T}{k}\Bigr)\\
& \ge \dfrac{k+i-1-a}{k}+a
\end{align}
Therefore we see that $\dfrac{k+i-1-a}{k}+a$ is the lower bound for $\sum\limits_{j=1}^{k}\lambda_j$, and $\sum\limits_{j=1}^{k}\lambda_j = \dfrac{k+i-1-a}{k}+a$ if and only if $T = \sum\limits_{j \, \in \, \mathcal{B}}\lambda_j = 0$ or 
\[
\begin{cases}{}
    & k\sum\limits_{j \, \in \, \mathcal{A}}(\lambda_j - 1) + k\sum\limits_{j \, \in \, \mathcal{C}}\lambda_j = k+i-1-a \\
    & \sum\limits_{j \, \in \, \mathcal{B}}\lambda_j = 0
\end{cases}
\]
The second constraint $\sum\limits_{j \, \in \, \mathcal{B}}\lambda_j = 0$ means that $\lambda_{j} = 0$ for all $j \, \in \, \mathcal{B}$. Therefore, all the nodes in $\mathcal{B}$ are untouched, and the Node capacity constraint $k\sum\limits_{j \, \in \, \mathcal{A}}(\lambda_j - 1) + k\sum\limits_{j \, \in \, \mathcal{C}}\lambda_j = k+i-1-a$ can be tightly satisfied just like the case $\boldsymbol{G}_i(n, k)$ when $n \ge k +i$ (recall that we have $|\mathcal{A}| = a \le i-1$, so after the Greedy matching allocation, we will be left with at least $(k+i-1) - (i-1) = k$ entirely unused nodes).

We have just proved that for any request vector $\boldsymbol{\lambda}_A$ such that:
\[
\begin{cases}{}
    & k\sum\limits_{j \, \in \, \mathcal{A}}(\lambda_j - 1)+ k\sum\limits_{j \, \in \, \mathcal{C}}\lambda_j + \sum\limits_{j \, \in \, \mathcal{B}}\lambda_j = k+i-1-a\\
    & \sum\limits_{j=1}^{k}\lambda_j = \dfrac{k+i-1-a}{k}+a
\end{cases}
\]
is servable. On the other hand, part 1 proved that any request vector $\boldsymbol{\lambda}_B$ such that
\[
\begin{cases}{}
    & k\sum\limits_{j \, \in \, \mathcal{A}}(\lambda_j - 1)+ k\sum\limits_{j \, \in \, \mathcal{C}}\lambda_j + \sum\limits_{j \, \in \, \mathcal{B}}\lambda_j = k+i-1-a \\
    & \sum\limits_{j=1}^{k}\lambda_j = i
\end{cases}
\]
is also servable. For any request vector \( \boldsymbol{\lambda} \) that tightly satisfies the Node capacity constraint (i.e., the first part of \eqref{eq:n=k+i-1:secondcase} holds), its sum rate \( \sum\limits_{j=1}^{k} \lambda_j \) must lie within the range  
\[
\Biggl[\frac{k+i-1-a}{k} + a, i\Biggr].
\]
Here, we have established that \( \frac{k+i-1-a}{k} + a \) serves as the lower bound for \( \sum\limits_{j=1}^{k} \lambda_j \), and Lemma~\ref{lem:sum_rate_bound} has proven that its upper bound is \( i \). Consequently, \( \boldsymbol{\lambda} \) must be servable, as it can be expressed as a linear combination of the two servable extreme servable points \( \boldsymbol{\lambda}_A \) and \( \boldsymbol{\lambda}_B \) (recall that the service polytope is convex).

\item None of the two constraints is tightly satisfied:
\begin{equation}
\begin{cases}
    & k\sum\limits_{j \, \in \, \mathcal{A}}(\lambda_j - 1) + k\sum\limits_{j \, \in \, \mathcal{C}}\lambda_j + \sum\limits_{j \, \in \, \mathcal{B}}\lambda_j < n-a = k+i-1-a \label{eq:n=k+i-1:thirdcase}\\
    & \sum\limits_{j=1}^{k}\lambda_j < i
\end{cases}
\end{equation}
In this case, we can always increase 1 request component $\lambda_j$ until either the first or second part of \eqref{eq:n=k+i-1:thirdcase} becomes equality (whichever comes first), and we come back to the previous cases (1) and (2). Thus, any vector that falls into this case can be served.
\end{enumerate}
\end{proof}
Fig.~\ref{mds_example} plots the SRRs of $\boldsymbol{G}_3(4, 3)$ and $\boldsymbol{G}_3(5, 3)$. The SRR of $\boldsymbol{G}_3(5, 3)$ is given by:
\[
\begin{aligned}
    &\begin{cases}
        \lambda_j \ge 0, \quad j \, \in \, [3], \\
        \sum\limits_{j=1}^3 \lambda_j \le 3, & \\
        3\lambda_1 + \lambda_2 + \lambda_3 \le 7, \quad    3\lambda_2 + \lambda_1 + \lambda_3 \le 7, \quad 3\lambda_3 + \lambda_2 + \lambda_1 \le 7, & \\
        3(\lambda_1 + \lambda_2) + \lambda_3 \le 9, \quad 3(\lambda_1 + \lambda_3) + \lambda_2 \le 9, \quad 3(\lambda_2 + \lambda_3) + \lambda_1 \le 9. &
    \end{cases}
\end{aligned}
\]
Note that the constraint \( \sum\limits_{j=1}^3 \lambda_j \leq 3 \) is tighter than the last three constraints, causing them inactive and redundant.
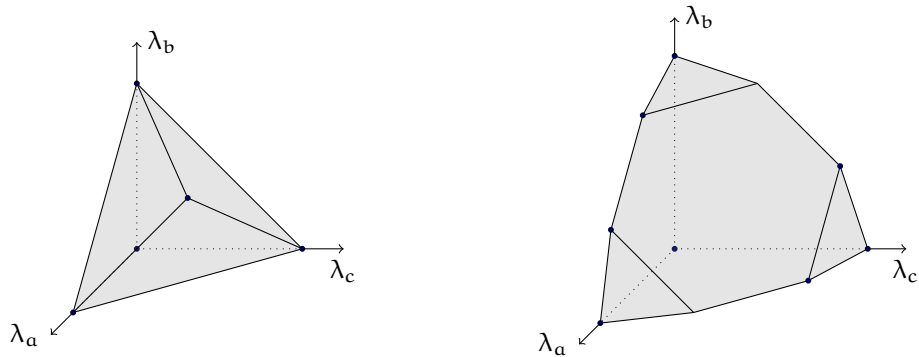
\begin{figure}
\centering  
\begin{tikzpicture}
[scale=1.05,
	back/.style={dotted},
	edge/.style={color=black},
	facet/.style={fill=gray,fill opacity=0.200000},
	vertex/.style={inner sep=0.5pt,circle,draw=blue!25!black,fill=blue!75!black,thick,anchor=base}]
    \begin{scope}[xshift = 0cm]
        \coordinate (0.00000, 0.00000, 0.00000) at (0.00000, 0.00000, 0.00000);
        \coordinate (0.00000, 0.00000, 2.00000) at (0.00000, 0.00000, 2.00000);
        \coordinate (0.00000, 2.00000, 0.00000) at (0.00000, 2.00000, 0.00000);
        \coordinate (1.00000, 1.00000, 1.00000) at (1.00000, 1.00000, 1.00000);
        \coordinate (2.00000, 0.00000, 0.00000) at (2.00000, 0.00000, 0.00000);
        \draw[edge,back] (0.00000, 0.00000, 0.00000) -- (0.00000, 0.00000, 2.00000);
        \draw[edge,back] (0.00000, 0.00000, 0.00000) -- (0.00000, 2.00000, 0.00000);
        \draw[edge,back] (0.00000, 0.00000, 0.00000) -- (2.00000, 0.00000, 0.00000);
        \node[vertex] at (0.00000, 0.00000, 0.00000)     {};
        \fill[facet] (1.00000, 1.00000, 1.00000) -- (0.00000, 0.00000, 2.00000) -- (0.00000, 2.00000, 0.00000) -- cycle {};
        \fill[facet] (2.00000, 0.00000, 0.00000) -- (0.00000, 0.00000, 2.00000) -- (1.00000, 1.00000, 1.00000) -- cycle {};
        \fill[facet] (2.00000, 0.00000, 0.00000) -- (0.00000, 2.00000, 0.00000) -- (1.00000, 1.00000, 1.00000) -- cycle {};
        \draw[edge] (0.00000, 0.00000, 2.00000) -- (0.00000, 2.00000, 0.00000);
        \draw[edge] (0.00000, 0.00000, 2.00000) -- (1.00000, 1.00000, 1.00000);
        \draw[edge] (0.00000, 0.00000, 2.00000) -- (2.00000, 0.00000, 0.00000);
        \draw[edge] (0.00000, 2.00000, 0.00000) -- (1.00000, 1.00000, 1.00000);
        \draw[edge] (0.00000, 2.00000, 0.00000) -- (2.00000, 0.00000, 0.00000);
        \draw[edge] (1.00000, 1.00000, 1.00000) -- (2.00000, 0.00000, 0.00000);
        \node[vertex] at (0.00000, 0.00000, 2.00000)     {};
        \node at (0,2) [left] {\small \textcolor{black}{2}};
        \node[vertex] at (0.00000, 2.00000, 0.00000)     {};
        \node at (2,0) [below] {\small \textcolor{black}{2}};
        \node[vertex] at (1.00000, 1.00000, 1.00000)     {};
        \node[vertex] at (2.00000, 0.00000, 0.00000)     {};
        \node at (-0.7,-0.7) [left] {\small \textcolor{black}{2}};
        \draw[-,thin] (0.7,0.65) -- (2.5,1) node[right, inner sep=2pt] {\small $(1, 1, 1)$};
        \draw[->] (2.00000, 0.00000, 0.00000) -- (2.50000, 0.00000, 0.00000)  node[below] {\small $\lambda_c$};
        \draw[->] (0.00000, 2.00000, 0.00000) -- (0.00000, 2.50000, 0.00000)  node[right] {\small $\lambda_b$};
        \draw[->] (0.00000, 0.00000, 2.00000) -- (0.00000, 0.00000, 2.90000)  node[left] {\small $\lambda_a$};
        \end{scope}

    \begin{scope}[xshift=6.5cm, yshift = 0cm]
	\coordinate (r) at (0.00000, 0.00000, 0.00000);
	\coordinate (a) at (0.00000, 0.00000, 2.33333);
	\coordinate (ab) at (0.00000, 2.00000, 1.00000);
	\coordinate (ba) at (0.00000, 1.00000, 2.00000);
	\coordinate (b) at (0.00000, 2.33333, 0.00000);
	\coordinate (ca) at (1.00000, 0.00000, 2.00000);
	\coordinate (ac) at (2.00000, 0.00000, 1.00000);
	\coordinate (bc) at (2.00000, 1.00000, 0.00000);
	\coordinate (cb) at (1.00000, 2.00000, 0.00000);
	\coordinate (c) at (2.33333, 0.00000, 0.00000);
	\draw[edge,back] (r) -- (a);
	\draw[edge,back] (r) -- (b);
	\draw[edge,back] (r) -- (c);
	\node[vertex] at (0.00000, 0.00000, 0.00000)     {};
	\fill[facet] (c) -- (ac) -- (ca) -- (a) -- (ba) -- (ab) -- (b) -- (cb) -- (bc) -- cycle {};	
	\draw[edge] (a) -- (ba) -- (ab) -- (b);
	\draw[edge] (a) -- (ca) -- (ac) -- (c);
	\draw[edge] (b) -- (cb) -- (bc) -- (c);
	\draw[edge] (ba) -- (ca);
	\draw[edge] (bc) -- (ac);
	\draw[edge] (cb) -- (ab);	
		
	\node[vertex] at (a)     {};
	\node[vertex] at (b)     {};
	\node[vertex] at (ac)     {};
	\node[vertex] at (bc)     {};
	\node[vertex] at (c)     {};
	\node[vertex] at (ab)     {};
	\node[vertex] at (ba)     {};
     \draw[-,thin] (0.7,1.2) -- (3.5,1.8) node[right, fill=gray!20, inner sep=2pt] {\small $\lambda_a + \lambda_b + \lambda_c = 3$};
	\draw[->] (c) -- (2.8, 0.00000, 0.00000)  node[below] {\small $\lambda_c$};
    \node at (0,2.3) [left] {\small \textcolor{black}{$\frac{7}{3}$}};
	\draw[->] (b) -- (0.00000, 2.8, 0.00000)  node[right] {\small $\lambda_b$};
    \node at (2.3, 0) [below] {\small \textcolor{black}{$\frac{7}{3}$}};
    \node at (1.95, 1.0) [right] {\small \textcolor{black}{(0, 1, 2)}};
    \node at (1, 2) [right] {\small \textcolor{black}{(0, 2, 1)}};
	\draw[->] (a) -- (0.00000, 0.00000, 3.3)  node[left] {\small $\lambda_a$};
    \node at (-0.85,-0.8) [left] {\small \textcolor{black}{$\frac{7}{3}$}};
        \end{scope}
\end{tikzpicture}
\caption{Service polytope of systematic MDS code $\boldsymbol{G}_3(4, 3)$ (left) and $\boldsymbol{G}_3(5, 3)$.}
    \label{mds_example}
\end{figure}
\subsection{Coding schemes with $k \le n < k+i-1$}
In this case, the general SRR is unknown. However, we prove that although Lemma~\ref{lem:sum_rate_bound} still holds, it can not be satisfied by too many vectors $\boldsymbol{\lambda}$. 
We see from Fig.~\ref{mds_example} that in $\mathcal{S}_3(5, 3)$, there is a plane of request vectors $\boldsymbol{\lambda}$ such that $\sum_{j=1}^k\lambda_j=3$, while in $\mathcal{S}_3(4, 3)$ there is only one such vector, $\boldsymbol{\lambda} = (1, 1, 1)$, that satisfies this constraint. In general, the SRR of $\boldsymbol{G}_i(k+i-1, k)$ contains a plane of service vectors $\boldsymbol{\lambda}$ such that $\sum\limits_{j=1}^k \lambda_j = i$.

We will prove that when $n < k+i-1$, there is only one vector $\boldsymbol{\lambda}$ that satisfies the constraint $\sum\limits_{j=1}^k \lambda_j = i$, namely $\boldsymbol{\lambda} = (1, 1, \hdots, 1, 0, 0, \hdots, 0)$ (with the first $i$ elements equal to 1 and the remaining elements equal to 0). The key idea behind the proof is that when $n$ is too small relative to $k+i-1$, the number of non-systematic nodes is insufficient to support any other allocation satisfying the given constraint. This result shows that in this case, apart from Greedy matching, no other rate-splitting scheme can serve this request vector.

\begin{theorem}\label{thm:greedy_optimal}
    For \( k \leq n \leq k+i-2 \), the service polytope \( \mathcal{S}_i(n, k) \) contains only the vector 
    \[
    \begin{array}{c}        
    \boldsymbol{\lambda} = (1, \hdots, 1, 0, \hdots, 0)
    \vspace{-6mm}\\
    \hspace{-6mm}
    \underbrace{\hspace{1.1cm}}_{\text{\( i \) elements}}
    \end{array}
    \]
    as the unique vector satisfying \( \sum\limits_{j=1}^k \lambda_j = i \). Moreover, any scheme other than Greedy matching cannot serve this request vector.
\end{theorem}
\begin{proof}
    Since \( \Gamma_G \) has one systematic node for each of the first \( i \) basis vectors, the vector \( \boldsymbol{\lambda} \) is servable, meaning it lies within the service polytope. Assume, toward a contradiction, that there exists another vector \( \boldsymbol{\mu} \neq \boldsymbol{\lambda} \) in the polytope such that \( \sum\limits_{j=1}^k \mu_j = i \). Theorem~\ref{thm:Greedy} guarantees that this vector is servable by Greedy matching. Applying Greedy matching to \( \boldsymbol{\mu} \), we obtain:
    \[
    \begin{cases}
        (\mu_j)^{\mathsf{s}} = 1, \quad \forall \, j \, \in \, \mathcal{A}, \\
        (\mu_j)^{\mathsf{s}} = \mu_j, \quad \forall \, j \, \in \, \mathcal{B}, \\
        (\mu_j)^{\mathsf{s}} = 0, \quad \forall \, j \, \in \, \mathcal{C}.
    \end{cases}
    \]
    After this step, the remaining request to be served is  
    \[
    \sum\limits_{j \, \in \, \mathcal{A}}(\mu_j - 1) + \sum\limits_{j \, \in \, \mathcal{C}}\mu_j,
    \]
    which must be handled by non-systematic recovery sets, each of cardinality \( k \).
    The number of non-systematic nodes in the system is \( n-i \). Each non-systematic recovery set contains at least two systematic nodes since \( k-(n-i) \geq 2 \). Therefore, for every portion \( \delta \) of  
\(
\sum\limits_{j \, \in \, \mathcal{A}}(\mu_j - 1) + \sum\limits_{j \, \in \, \mathcal{C}}\mu_j
\)
served by non-systematic recovery sets, at least \( 2\delta \) must be drawn from the systematic nodes in \( \mathcal{B} \), as all nodes in \( \mathcal{A} \) have already been fully utilized. In other words, to serve an amount \( \delta \) in the remaining requests, at least \( 2\delta \) must be drawn from  
\(
\sum\limits_{j \, \in \, \mathcal{B}}(1-\mu_j).
\)

    On the other hand, from the condition  
    \[
    i = \sum\limits_{j=1}^{k}\lambda_j = \sum\limits_{j \, \in \, \mathcal{A} \,\cup\, \mathcal{B} \,\cup\, \mathcal{C}} \lambda_j,
    \]
    it follows that  
    \[
    \sum\limits_{j \, \in \, \mathcal{A}}(\lambda_j - 1) + \sum\limits_{j \, \in \, \mathcal{C}}\lambda_j = \sum\limits_{j \, \in \, \mathcal{B}}(1-\lambda_j).
    \]
    This means that after the Greedy matching allocation, the sum of the remaining requests to be served must equal the total free capacity of the systematic nodes. However, we have previously shown that to serve an amount \( \delta \) in the remaining requests, at least \( 2\delta \) must be drawn from  
    \(
    \sum\limits_{j \, \in \, \mathcal{B}}(1-\mu_j).
    \)
    This can happen only if $\sum\limits_{j \, \in \, \mathcal{A}}(\lambda_j - 1) + \sum\limits_{j \, \in \, \mathcal{C}}\lambda_j = \sum\limits_{j \, \in \, \mathcal{B}}(1-\lambda_j) = 0$ or equivalently,
    \[
    \begin{cases}
        \lambda_j = 1, \quad \forall \, j \, \in \, \mathcal{A}, \\
        \lambda_j = 0, \quad \forall \, j \, \in \, \mathcal{B} \,\cup\, \mathcal{C}.
    \end{cases}
    \]
    This implies that  
    \[
    \begin{array}{c}        
    \boldsymbol{\mu} = (1, \hdots, 1, 0, \hdots, 0) = \boldsymbol{\lambda},
    \vspace{-6mm}\\
    \hspace{-16mm}
    \underbrace{\hspace{1.05cm}}_{\text{\( i \) elements}}
    \end{array} 
    \]
    contradicting our assumption that \( \boldsymbol{\mu} \neq \boldsymbol{\lambda} \). 

    Thus, no such vector \( \boldsymbol{\mu} \) can exist, proving that \( \boldsymbol{\lambda} \) is the unique vector satisfying \( \sum\limits_{j=1}^k \lambda_j = i \).

    Moreover, the previous capacity argument also shows that $\boldsymbol{\lambda}$ can not be served by any scheme other than Greedy matching. In other words, there is only one matching $\boldsymbol{w}$ in the matching polytope such that $\boldsymbol{\lambda} = \boldsymbol{\lambda}(\boldsymbol{w})$.
\end{proof}
\section{Conclusions and Future Work}\label{Sec:Conclusion}
We presented a rigorous analysis of the \textit{service rate region} (SRR) for distributed storage systems that employ MDS codes. We used graph-theoretic methods to comprehensively characterize achievable rates. By constructing a family of MDS generator matrices with varying numbers of systematic columns, we showed that increasing the number of systematic columns in the generator matrix of the same code strictly enlarges the SRR. We introduced two bounding simplices, the \textbf{Maximal Matching} and \textbf{Maximal Achievable} simplices, which provide clear geometric boundaries on feasible request rates. A key technical contribution was the proposal of a \textbf{Greedy Matching} allocation strategy, for which we proved optimality in several scenarios, thereby demonstrating that specific extreme points of the SRR boundary are attained only via Greedy allocation. Using this scheme, we developed the first explicit characterizations of MDS-coded SRRs under various configurations of \(n\) (servers), \(k\) (data objects), and \(i\) (systematic servers). We anticipate that these results will guide both practical code design and broader theoretical explorations of the interplay between code parameters and data-access performance.

In particular, we showed that the SRR for this class of codes strictly increases with the number of systematic columns, and thus fully systematic MDS codes achieve the largest SRR. Moreover, we have previously observed that limited replacement of some coded nodes with replicated systematic nodes can enlarge the SRR. However, these results should not discourage the study of non-systematic or partially systematic codes. In many practical settings, storing raw data on specific servers may be prohibited due to privacy, security, or architectural constraints, making non-systematic designs unavoidable. Furthermore, distributed storage systems often balance multiple objectives beyond maximizing service capacity, such as reducing data-access cost, supporting access asynchrony, or providing data hiding. From a theoretical standpoint, analyzing the entire family of codes, including non-systematic ones, is essential for developing general tools and insights for SRR characterization. Thus, understanding non-systematic codes remains fundamental both for practical system design and for establishing a unified theoretical framework for service rate analysis.

Until now, only a few general techniques have been used to derive SRRs; e.g., graph-theoretic methods were used for MDS codes because of their combinatorial properties, and geometric methods were used for Reed--Muller codes because of their geometric interpretations. Simplex codes can be seen from both perspectives \cite{Service:journals/tit/AktasJKKS21,Service:preprint/arxiv/LySL25}. A natural direction for future work is to extend SRR analysis to broader classes of storage codes, including regenerating and locally recoverable codes, whose recovery constraints differ fundamentally and may require new analytical techniques. Another promising direction is to leverage the methods developed for SRR characterization to design coding schemes tailored to achieve a desired, predefined SRR.

\section*{Acknowledgment}
This research was partially supported by the National Science Foundation under Grant No. CIF-2122400. The authors thank V.~Lalitha for her essential comments on an earlier draft of this work, and the anonymous reviewers for their valuable input. 

\appendix

\section{Dual Relationship between Fractional Matching and Vertex Cover}
\label{app:duality}

In this appendix, we explicitly formulate the fractional matching problem and the fractional vertex cover problem as a pair of Primal and Dual Linear Programs (LP). We then invoke the Strong Duality Theorem to establish the equality $\nu^*(\Gamma) = \tau^*(\Gamma)$ used in Section~\ref{Sec:Problem_Formulate}-D.

Consider the recovery hypergraph $\Gamma = (V, E)$ with incidence matrix $\boldsymbol{A} \in \{0,1\}^{|V| \times |E|}$, where $A_{v, \epsilon} = 1$ if vertex $v \in \epsilon$ and $0$ otherwise.

\subsection{The Primal Problem: Fractional Matching}
The fractional matching number $\nu^*(\Gamma)$ is defined as the maximum size of a fractional matching. This is an optimization problem where we seek to maximize the sum of edge weights $\boldsymbol{w} \in \mathbb{R}_{\ge 0}^{|E|}$ subject to vertex capacity constraints. Formally, the Primal LP is:

\begin{align}
    (\text{Primal}) \quad \nu^*(\Gamma) = \quad & \text{maximize} \quad \sum_{\epsilon \in E} w_\epsilon \nonumber \\
    & \text{subject to} \quad \sum_{\epsilon \ni v} w_\epsilon \le 1, \quad \forall\, v \in V, \nonumber \\
    & \phantom{\text{subject to}} \quad w_\epsilon \ge 0, \quad \forall\, \epsilon \in E.
\end{align}

In matrix notation, letting $\boldsymbol{1}_{|E|}$ and $\boldsymbol{1}_{|V|}$ denote all-one column vectors of appropriate dimensions, the Primal problem is:
\[
 \text{max } \{ \boldsymbol{1}_{|E|}^\top \boldsymbol{w} : \boldsymbol{A}\boldsymbol{w} \le \boldsymbol{1}_{|V|},\ \boldsymbol{w} \ge \boldsymbol{0} \}.
\]

\subsection{The Dual Problem: Fractional Vertex Cover}
To derive the Dual LP, we associate a dual variable $d_v$ with each constraint in the Primal (i.e., for each vertex $v \in V$). Standard LP duality transformations dictate that~\cite{Linear_programming}:
\begin{itemize}
    \item The maximization objective becomes minimization.
    \item The right-hand side vector of the Primal ($\boldsymbol{1}_{|V|}$) becomes the cost vector of the Dual objective function.
    \item The Primal coefficient matrix $\boldsymbol{A}$ is transposed to $\boldsymbol{A}^\top$.
    \item The inequality direction reverses ($\le$ becomes $\ge$).
\end{itemize}

Thus, the Dual LP is given by:
\begin{align}
    (\text{Dual}) \quad \tau^*(\Gamma) = \quad & \text{minimize} \quad \sum_{v \in V} d_v \nonumber \\
    & \text{subject to} \quad \sum_{v \in \epsilon} d_v \ge 1, \quad \forall\, \epsilon \in E, \nonumber \\
    & \phantom{\text{subject to}} \quad d_v \ge 0, \quad \forall\, v \in V.
\end{align}

In matrix notation, this is:
\[
 \text{min } \{ \boldsymbol{1}_{|V|}^\top \boldsymbol{d} : \boldsymbol{A}^\top \boldsymbol{d} \ge \boldsymbol{1}_{|E|}, \boldsymbol{d} \ge \boldsymbol{0} \}.
\]
The constraints of the Dual problem require that for every edge $\epsilon$, the sum of the weights of the vertices incident to $\epsilon$ is at least 1. This is precisely the definition of a fractional vertex cover, and the objective minimizes the total weight of such a cover.

\subsection{Strong Duality}
The constraints for both the Primal and Dual problems define non-empty, bounded polytopes (since all weights are non-negative and bounded by the number of vertices or edges). Therefore, optimal solutions exist for both. The \textit{Strong Duality Theorem} of Linear Programming states that if the primal problem has an optimal solution, then the dual problem also has an optimal solution, and the optimal objective values are equal~\cite{Linear_programming}. Consequently:
\[
    \nu^*(\Gamma) = \tau^*(\Gamma).
\]
This equality allows us to upper-bound the sum rate of any achievable request vector $\boldsymbol{\lambda}$ using the minimum weight of a fractional vertex cover, as detailed in Proposition~\ref{prop:sum_bound}.
\bibliography{bibliography}
\bibliographystyle{IEEEtran}

\end{document}